\documentclass[pra,superscriptaddress,twocolumn,10pt,times, ]{revtex4-1}
\pdfoutput=1

\usepackage[utf8]{inputenc}
\usepackage[english]{babel}
\usepackage[T1]{fontenc}

\usepackage[unicode=true,bookmarks=true,bookmarksnumbered=false,bookmarksopen=false,breaklinks=false,pdfborder={0 0 1}, backref=false,colorlinks=true]{hyperref}

\usepackage{bbm, stmaryrd, tabu, multirow, float, lipsum, amsmath, amsthm, amssymb, latexsym, mathrsfs, hhline, graphicx, braket, enumitem}
\usepackage[autostyle]{csquotes}
\usepackage[caption=false]{subfig}
\usepackage{tikz-cd}
\usetikzlibrary{decorations.pathreplacing,calligraphy}
\usetikzlibrary{backgrounds}

\usepackage{xcolor}
\definecolor{myred}{rgb}{0.8,0.3,0.2}
\definecolor{myredfill}{rgb}{1,0.93,0.93}
\definecolor{myred2}{rgb}{0.6,0.5,0.3}
\definecolor{myredfill2}{rgb}{0.93,0.83,0.83}
\definecolor{myblue}{rgb}{0.05,0.25,0.7}
\definecolor{mybluefill}{rgb}{0.92,0.98,1}
\definecolor{mygreen}{rgb}{0.3,0.6,0.4}
\definecolor{mygreenfill}{rgb}{0.93,1,0.97}
\definecolor{mygrey}{rgb}{0.4,0.4,0.4}
\definecolor{mygreyfill}{rgb}{0.95,0.95,0.95}
\definecolor{mypurple}{rgb}{0.6,0.2,0.6}
\definecolor{mypurplefill}{rgb}{0.99,0.94,0.99}
\definecolor{myyellow}{rgb}{9,0.9,0.3}
\definecolor{myyellowfill}{rgb}{1,1,0.92}
\definecolor{myorange}{rgb}{9,0.7,0.3}
\definecolor{myorangefill}{rgb}{1,0.95,0.92}
\definecolor{mybrown}{rgb}{0.4,0.25,0.1}
\definecolor{mybrownfill}{rgb}{1,0.9,0.8}
\definecolor{mediumgrey}{rgb}{0.7,0.7,0.7}
\definecolor{linkred}{rgb}{0.6,0.1,0.1}
\definecolor{citeblue}{rgb}{0.2,0.35,0.75}
\definecolor{urlblue}{rgb}{0.2,0.25,0.45}

\hypersetup{
     linkcolor = linkred,
     anchorcolor = linkred,
     citecolor = citeblue,
     filecolor = linkred,
     urlcolor = urlblue
     }

\newtheorem{theorem}{Theorem}

\newtheorem{lemma}{Lemma}
\newtheorem*{lemma*}{Lemma}

\newtheorem{proposition}{Proposition}

\def\ketbra#1#2{\ket{#1\vphantom{#2}}\!\bra{#2\vphantom{#1}}}
	
\begin{document}

\title{Quantifiers of Noise Reducibility Under Restricted Control}

\author{Graeme D. Berk}
\affiliation{Nanyang Quantum Hub, School of Physical and Mathematical Sciences, Nanyang Technological University, Singapore 639673}

\author{Kavan Modi}
\affiliation{Science, Mathematics and Technology Cluster, Singapore University of Technology and Design, 8 Somapah Road, 487372 Singapore}
\affiliation{School of Physics and Astronomy, Monash University, Victoria 3800, Australia}

\author{Simon Milz}
\email{S.Milz@hw.ac.uk}

\affiliation{Institute of Photonics and Quantum Sciences, School of Engineering and Physical Sciences, Heriot-Watt University, Edinburgh EH14 4AS, United Kingdom}
\affiliation{School of Physics, Trinity College Dublin, Dublin 2, Ireland}
\affiliation{Trinity Quantum Alliance, Unit 16, Trinity Technology and Enterprise Centre, Pearse Street, Dublin 2, D02YN67, Ireland}

\date{\today}

\begin{abstract}
The correlation structure of multitime quantum processes -- succinctly described by quantum combs -- is an important resource for many quantum information protocols and control tasks. Inspired by approaches for quantum \textit{states}, we introduce quantifiers of the practical utility of quantum processes that satisfy monotonicity properties, thus overcoming shortcomings in previous state-motivated approaches.  Applying these quantifiers to the problem of noise reduction of a quantum process under open-loop control, they are shown to represent the largest amount of temporal mutual information that a process can possibly exhibit. In addition, we study their resource composition behaviour and connect them to the recently introduced notion of generalised comb divergences. Finally, in light of these new quantifiers, we re-interpret the numerical findings of Ref.~\cite{extractingdynamicalquantumresources} on the relationship of dynamical decoupling and non-Markovian memory, which were based on insufficient resource quantifiers, and show that its main conclusion -- the interpretation of dynamical decoupling as a resource distillation -- still holds.
\end{abstract}

\maketitle
\section{Introduction}
Quantum technologies like quantum computers require exquisite control over quantum systems, which is a daunting scientific challenge. A key source of difficulty towards this goal is designing hardware that sufficiently isolates quantum systems from uncontrolled interactions with their environment. To reduce the burden that this hardware limitation places on quantum devices, there are a wealth of noise reduction methods, such as quantum error correction~\cite{PhysRevA.52.R2493,RevModPhys.87.307,Knill_2005}, quantum error mitigation~\cite{PhysRevLett.119.180509,PhysRevX.8.031027,Kandala_2019}, and dynamical decoupling~\cite{dynamicaldecouplingofopenquantumsystems,demonstrationoffidelityimprovement,dynamicaldecouplingofasingleelectronspinatroom,preservingelectronspincoherence,tong2025empiricallearningdynamicaldecoupling,PhysRevLett.134.050601,PhysRevApplied.20.064027} (DD). DD involves applying a rapid sequence of unitary pulses to the system of interest $s$, designed to `average away' the influence the environment $e$ has on $s$, leading to an overall quantum channel $\mathcal{T}$ that is as noiseless as possible. Among noise reduction methods, it is particularly practical because it requires no measurement feedback or post-processing, nor does it require additional redundant qubits. However, this latter feature makes it somewhat unclear what resource DD is actually leveraging to reduce noise. Similarly, it raises the question of which quantum processes are amenable to DD, and to what extent they can be de-noised.

The most common descriptions of noisy quantum processes are quantum channels and Markovian master equations. In a generic repeated interactions scenario -- like the one encountered in DD, where an experimenter applies control sequences to denoise a process -- these descriptions make a Markov approximation~\cite{Breuer:2007juk} that the actions of the experimenter on $s$ cannot influence $e$. While this assumption is often warranted, and leads to a significant complexity reduction in the description of the prevalent temporal correlations, doing so may discard potentially useful information about $s$~\cite{nonmarkovianquantumprocesses, quantumstochasticprocesses}. Moreover, multitime correlations are not captured by the Markov approximation, and disregarding them may lead to an overestimation of randomness/noise/irreversibility~\cite{regularitiesunseenrandomnessobserved}, and yield a mischaracterisation of the dynamics of the system altogether~\cite{reduceddynamicsneednotbe}. Indeed, the noise observed in real-world quantum processors is often non-Markovian~\cite{quantifyingnonmarkovianityinasuperconditing, White_2020,unifyingnonmarkoviancharacterisationefficient}, making this a significant practical concern. Employing a fully non-Markovian description of quantum processes thus has a twofold advantage: First, it enables the correct prediction of the general evolution of a quantum system subject to noise and external control. Second, just like in the spatial case, \textit{temporal} quantum correlations that present as complex memory offer a potential resource, and are expected to yield enhanced performance at a wide range of quantum control tasks~\cite{strategiccodeunifiedspatiotemporal,quantumprocessesthermodynamicresources} like, for example, DD. Properly gauging such enhancement requires a complete description of memory effects and the derivation of quantifiers of their resourcefulness.

A fully-fledged framework that can account for \textit{all} possible memory effects in quantum processes has been developed previously in terms of the process tensor~\cite{nonmarkovianquantumprocesses,quantumstochasticprocesses,anintroductiontooperational}, and its capability to experimentally characterise and diagnose general quantum noise has been demonstrated~\cite{PRXQuantum.3.020344, White_2020, unifyingnonmarkoviancharacterisationefficient}. A process tensor is a complete characterisation of a quantum system's observable statistics, given any way that it can be accessed at a discrete number of times. Thus, it is natural to study process tensors as a vehicle to quantify the amenability of quantum systems to discrete-time quantum control tasks in a manner that accounts for \textit{all} potential resources.

Similar quantum resource questions have been considered in the literature when investigating the utility of quantum \textit{states} to an experimenter seeking to perform some operational task. For these cases, relative entropy to a set of `free' (i.e., useless) states~\cite{review} is a highly significant and mathematically well-behaved quantifier~\cite{Berta_2023, hayashi2025generalizedquantumsteinslemma}. Drawing inspiration from the quantum state case, it appears natural to also quantify the usefulness of quantum \textit{processes} (defined by process tensors) in terms of a relative entropy with respect to a set of processes that are considered to be `free/useless' under specified operational constraints. For example, the amenability of a process to DD appears to be closely linked to its non-Markovian correlations~\cite{nonmarkoviannoisethatcannotbe, dynamicaldecouplingefficiencyversusquantumnonmarkovianity, Gough_2017, extractingdynamicalquantumresources, PRXQuantum.2.040351, Figueroa-Romero_2024} (albeit this relationship is far from being fully understood), seemingly making the relative entropy of a process with respect to the set of Markovian processes an obvious marker of its amenability to DD. 

However, generalising state-based resource quantifiers to quantum processes requires the careful treatment of subtle mathematical differences between states and processes. In Ref.~\cite{extractingdynamicalquantumresources} we sought to use the framework of independent quantum instruments (IQI) -- where the experimenter can perform arbitrary quantum control operations at fixed times, without memory between them -- to evaluate the amenability of processes to DD. To achieve this, we employed the relative entropy of a multitime process to its `natural' memoryless version as a quantifier and claimed that it decreases monotonically under DD operations. This latter property mirrors the requirement that such a resource quantifier should appraise the resources of a process that can be exploited in DD, and should thus not be increased by the corresponding control operations; put differently, it should not be possible to increase the value of a process for a task by means of `free' operations that can be performed. Later, Ref.~\cite{zambon2024processtensordistinguishabilitymeasures} found that the `proof' of monotonicity of relative entropy under the IQI scenario in Ref.~\cite{extractingdynamicalquantumresources} had a subtle mistake, and provided explicit counterexamples, implying that relative entropy on its own cannot meaningfully characterise the utility of non-Markovianity for DD in general.

In this paper, we close this identified gap by first deriving a class of resource quantifiers for multitime processes for arbitrary control tasks that are still rooted in relative entropy, but are monotonic under the control allowed in the respective setups. Narrowing our focus to the DD problem of noise reduction on a quantum process that can be controlled at initial, intermediate, and final times $\hat{n}:=\{t_{\text{i}},t_1,\dots,t_n,t_{\text{f}} \}$, in Sec.~\ref{sec:monotones} we introduce three temporal correlation quantifiers $\overline{I}_{\hat{m}},\overline{M}_{\hat{m}},\overline{N}_{\hat{m}}$ that measure total, Markovian, and non-Markovian correlations of a process, respectively for sets of times $\hat{m} \subseteq \hat{n}$. We prove that these are indeed monotonic, and hence meaningful resource quantifiers for quantum processes. For the case $\hat{m} = \emptyset$, where DD operations have been performed at \textit{all} available times, $I_\emptyset$ corresponds to the input-output correlations of the resulting channel $\mathcal{T}_\emptyset$, thus quantifying the maximal possible success of DD. 

More generally, one may desire to still retain access to the process at intermediate times, in which case the quantifiers $\overline{I}_{\hat{m}},\overline{M}_{\hat{m}}$ and $\overline{N}_{\hat{m}}$ evaluate the different types of correlations that can be present at these lowered temporal resolutions. Detailing the relationship between these quantities, in Sec.~\ref{sec:subadditivity} we prove the subadditivity relation $\overline{I}_{\hat{m}} \leq \overline{M}_{\hat{m}}+\overline{N}_{\hat{m}}$, which is a manifestation of the fact that the distinct types of correlations corresponding to each of the quantifiers $\overline{I}_{\hat{m}},\overline{M}_{\hat{m}}$ and $\overline{N}_{\hat{m}}$, respectively, cannot be simultaneously optimised in general. As a consequence, this subaddititivity condition shows that optimally harnessing Markovian correlations typically requires sacrificing some non-Markovian correlations and vice versa.

The reduction in the number of intermediate times inherent to these quantifiers is critical to explaining the fact that DD can reduce noise on a single copy of a system under experimental control that cannot increase resource value. Indeed, seemingly, DD makes a process more valuable (that is, less noisy) by only using operations that are deemed free (the operations of IQI, used in DD). We resolve this apparent contradiction and explain DD as a kind of resource distillation among subsystems of a process tensor corresponding to the \textit{same system}, but at \textit{different times}. This is in contrast to quantum error correction, which can be understood as resource distillation among spatially distinct copies of a system. These two kinds of manipulations can be -- and have been~\cite{Ng_2011,optimallycombining,han2024protectinglogicalqubitsdynamical} -- combined to leverage both effects simultaneously. The IQI scenario encompasses these cases (so long as the error correction scheme does not have memory between steps), thus motivating the study of how the quantifiers we introduce behave when we have access to multiple copies of a quantum process. In Sec.~\ref{sec:compositionstructure}, we study the behaviour of our quantifiers under both sequential and parallel composition of processes, corresponding to having access to multiple processes one after the other or simultaneously (see Fig.~\ref{fig::composition}). We find that our quantifiers are additive and superadditive under sequential and parallel composition, respectively. 

In Sec.~\ref{sec:correspondence} we compare the introduced quantifiers to generalised comb divergences~\cite{zambon2024processtensordistinguishabilitymeasures}, an operationally meaningful notion of distance between quantum processes obtained through an optimisation over all possible `complementary' combs. In particular, we show that the quantifiers we provide are upper bounded by a specific type of comb divergence based on optimisation only over combs that are reachable within the operational constraints of the scenario.  

Finally, in Sec.~\ref{sec:numerics} we study the algorithm for multitimescale optimal dynamical decoupling (MODD), introduced in Ref.~\cite{extractingdynamicalquantumresources} to find optimal decoupling sequences for given non-Markovian processes. We show that this algorithm can be interpreted as an approximate optimisation method for quantifier $\overline{I}_{\emptyset}$, which quantifies the noiselessness of the resulting dynamics. With this understanding in mind, we observe from the numerical results of Ref.~\cite{extractingdynamicalquantumresources} that traditional DD methods only extract a small portion of the available resources from a process.

Together, our results provide faithful and mathematically well-behaved resource quantifiers for multi-time processes at all temporal resolutions that can be adapted to any considered control task. In particular, they provide the possible limits of denoising via DD, allow its interpretation as a resource distillation, and demonstrate the trade-off between different resources present in quantum processes.

\section{Framework}
\subsection{Process tensors}
\begin{figure}
\centering
\includegraphics[width=0.99\linewidth]{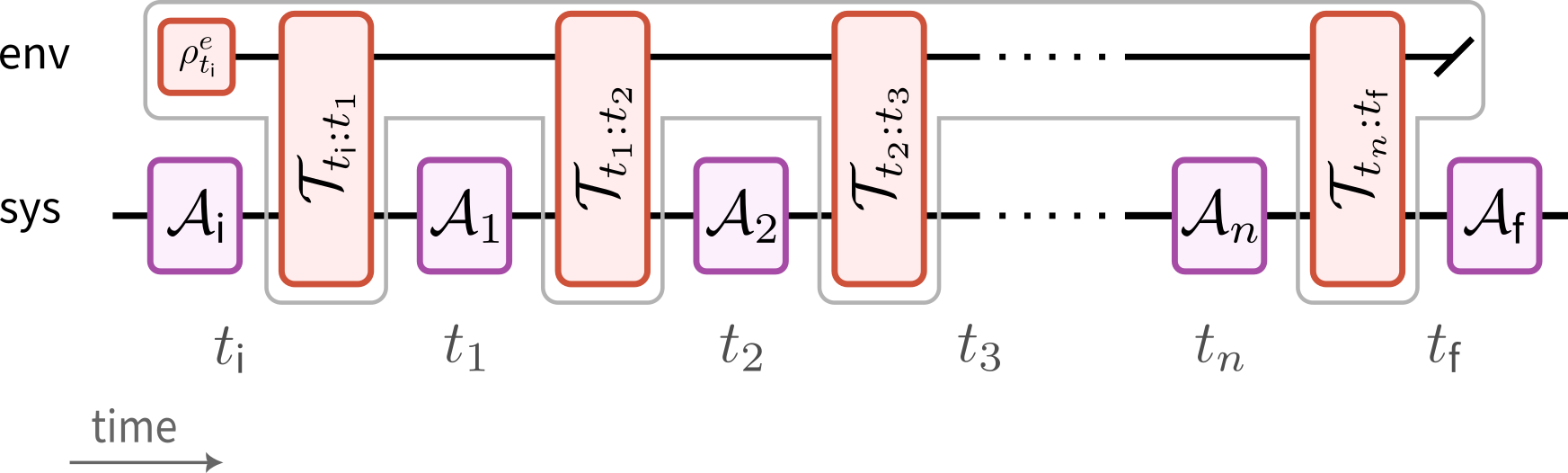}
\caption{\textbf{Multi-time quantum process.} Throughout, we consider the dynamics of a system $s$ coupled to an environment $e$ initially in state $\rho_\texttt{i}^e$. The system can be accessed/manipulated (depicted here by the maps $\mathcal{A}_1, \mathcal{A}_2, \dots, \mathcal{A}_n$) at times $\hat{n} =\{t_1, \dots, t_n\}$ in addition to the initial and final times $t_{\text{i}}$ and $t_{\text{f}}$, respectively. The evolution of the system and environment in between interventions is given by $\mathcal{T}$, and can differ for different time steps. Through the interaction with the environment, the process can display complex, non-Markovian correlations in time. The resulting process tensor description of the process is depicted by the grey outline. In this work, we are predominantly concerned with the IQI scenario (see Sec.~\ref{sec:restrictedcombs}), where all $\mathcal{A}_k$ are independent. More generally, they can be correlated or consist of different layers of operations that map between different levels of temporal resolution (see Figs.~\ref{fig:superprocess} and~\ref{fig:processtensorsuperprocess}).}
\label{fig::sys_env_dyn}
\end{figure}
Consider the dynamics of an open $d_s$-dimensional quantum system $s$ (interacting with an environment $e$) that is accessible to an experimenter at initial time $t_{\text{i}}$ and final time $t_{\text{f}}$, plus a set of intermediate times $\hat{n}=\{ t_1 , \dots , t_n \}$ (see Fig.~\ref{fig::sys_env_dyn}). This setup corresponds, for example, to dynamical decoupling, which consists of predetermined sequences of unitary pulses -- typically chosen from a decoupling group~\cite{dynamicaldecouplingofopenquantumsystems} -- at times $\hat{n}$ with the aim to average away the noise induced by the $se$ interaction. In general, as $s$ and $e$ interact, the experimenter's actions on $s$ may lead to changes in $e$, which in turn can influence $s$ again at a later time, manifesting as temporal correlations, more commonly known as \textit{non-Markovian} behaviour.
\begin{figure}[ht!]
\includegraphics[width=0.97\linewidth]{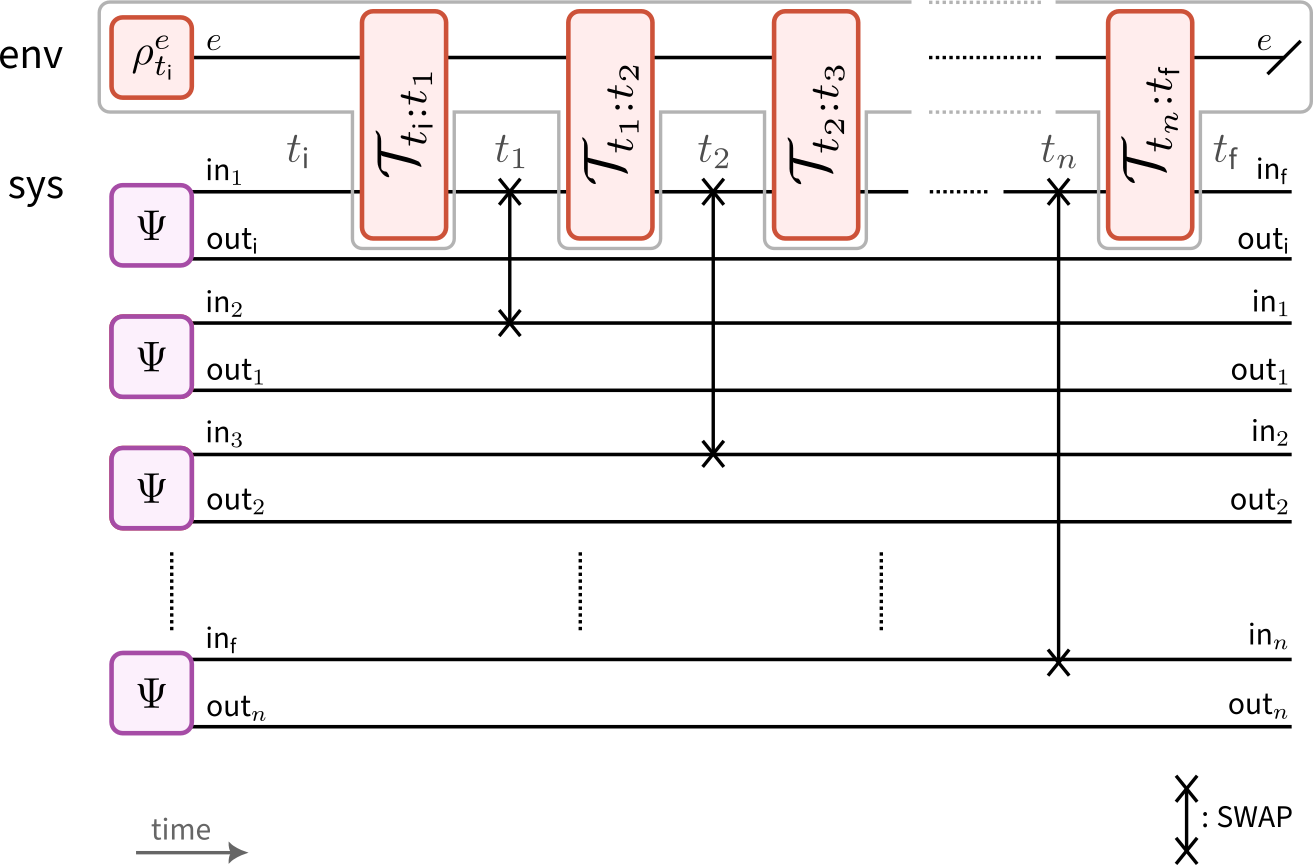}
\caption{\textbf{Open System Dynamics and Choi states.} Throughout, we consider the scenario of a system $s$ of interest that interacts with an environment $e$ and can be probed at times $t_\text{i}, t_1, \dots, t_n, t_\text{f}$. To map the process tensor (grey outline) that fully describes this multi-time setup onto a quantum state, half of a maximally entangled state is `fed in' (represented using circuit notation by the SWAP operator) at each of the corresponding times, resulting in the $2n+2$-partite state $\mathbf{T}_{\hat{n}}$ of Eq.~\eqref{eq:quantumprocessdef}.}
\label{fig::Choi}
\end{figure}

When the experimental interventions are brief compared to the time intervals between them (and compared to typical timescales of the dynamics), the influences of the environment and the experimenter cleanly separate within the process tensor/quantum comb formalism~\cite{theoreticalframeworkforquantumnetworks, nonmarkovianquantumprocesses,quantumstochasticprocesses,anintroductiontooperational}. The process tensor/quantum comb $\mathbf{T}_{\hat{n}}$ of a multitime process (i.e., a process that is probed at multiple points in time) encodes its statistical response to any way that the process may be probed by an experimenter. In particular, it accounts for all observable non-Markovian correlations. These correlations can most easily be described via the Choi state representation of the process tensor~\cite{anintroductiontooperational}, generalising channel-state duality. Specifically, the Choi state of a multitime process is obtained by letting it act on half of a maximally entangled state at each of the considered times (see Fig.~\ref{fig::Choi}).

Let the dynamics of system and environment $se$ from any $t_k$ to $t_{k+1}$ be given by a map $\mathcal{T}_{t_k:t_{k+1}}: \mathcal{L}(\mathcal{H}_{\text{in}_{k+1}} \otimes \mathcal{H}_{e} \rightarrow \mathcal{H}_{\text{in}_{k+1}} \otimes \mathcal{H}_{e})$. The $2n+2$-partite Choi state of the corresponding quantum process is
\begin{equation} \label{eq:quantumprocessdef}
\begin{aligned}
   \mathbf{T}_{\hat{n}} 
   =&  {\text{tr}}_{e} \Big\{   \bigcirc_{t_k \in t_{\text{i}}  \cup \hat{n} }       \mathcal{T}_{t_k:t_{k+1}} \Big[ \bigotimes_{t_k \in t_{\text{i}}  \cup \hat{n}}  \Psi^+_{{\text{in}}_{k+1}{\text{out}}_{k}}   \otimes  \rho_{t_{\text{i}}}^{e}  \Big] \Big\},
\end{aligned}
\end{equation}
where $\rho_{t_{\text{i}}}^{e}$ is the initial state of the environment and $\Psi^+_{{\text{in}}_{k+1}{\text{out}}_{k}}$ is the maximally entangled state on $\mathcal{H}_{{\text{in}}_{k+1}} \otimes \mathcal{H}_{\text{out}_k}$, with $\mathcal{H}_{{\text{in}}_{k+1}} \cong \mathcal{H}_{\text{out}_k}$.

Following the above definition, we have $\mathbf{T}_{\hat{n}} \in \mathcal{L}(\mathcal{H}_{\text{in}_\text{f}} \otimes \mathcal{H}_{\text{out}_n} \otimes \mathcal{H}_{\text{in}_n} \otimes \cdots \mathcal{H}_{\text{in}_1} \otimes \mathcal{H}_{\text{out}_\text{i}}$), where each pair $\mathcal{H}_{\text{in}_k} \otimes \mathcal{H}_{\text{out}_k}$ corresponds to the `slot' at time $t_k \in \hat{n}$ (see Fig.~\ref{fig::Choi}). The interaction with the environment $e$ may lead to correlations/memory in the resulting Choi state $\mathbf{T}_{\hat{n}}$. If it does not, then the process is Markovian, and all intermediate maps between times are of the form 
$\mathcal{T}^{s}_{t_k:t_{k+1}}$, i.e., they act on the system alone, leading to a process tensor of product form~\cite{nonmarkovianquantumprocesses}:
\begin{gather}
\label{eq:Markovian_process}
   \mathbf{T}^{\mathscr{M}}_{\hat{n}} = \mathbf{T}_{t_n} \otimes \mathbf{T}_{t_{n-1}} \otimes \cdots \otimes \mathbf{T}_{t_1} \otimes \mathbf{T}_{t_\text{i}},
\end{gather}
where $\mathbf{T}_{t_n} \in \mathcal{L}(\mathcal{H}_{\text{in}_{k+1}} \otimes \mathcal{H}_{\text{out}_k})$ is the Choi state of $\mathcal{T}^{s}_{t_k:t_{k+1}}$. Markovian processes are represented by Choi states that have no correlations between spaces that correspond to non-adjacent times, although there can still be correlations between the adjacent spaces $\text{in}_{k+1}$ and $\text{out}_k$ in each of the $\mathbf{T}_{t_k}$ (below we will refer to this type of correlations as Markovian information). On the other hand, non-Markovianity corresponds to the presence of inter-step correlations within processes $\mathbf{T}_{\hat{n}}$ that cannot be written in the form of Eq.~\eqref{eq:Markovian_process}. 

\subsection{Control combs}
In a similar vein, the experimental control itself can be represented by a process tensor/quantum comb, recalling that we allow for experimental intervention at times $ t_{\text{i}} \cup \hat{n} \cup t_{\text{f}}$. The corresponding interventions could be independent (the case we will predominantly consider in this work, see Fig.~\ref{fig::sys_env_dyn}), or correlated, both classically or quantum (see Fig.~\ref{fig:superprocess}). In full generality, a \textit{control comb} $\mathbf{Z}_{\hat{n} \hat{m}}$ corresponding to the experimental interventions, can transform process tensors with $\hat{n}$ intermediate times, to ones with $\hat{m} \subseteq \hat{n}$ intermediate times. For example, in Fig.~\ref{fig::sys_env_dyn} we have a control comb $\mathbf{Z}_{\hat{n} \emptyset}$ that maps a comb with ${\hat{n}}$ available intermediate times to a channel with no remaining intermediate times/slots. On the other hand, the control comb in Fig.~\ref{fig:superprocess} transforms a comb with intermediate times $\{t_1, t_2, t_3\}$ to one with a single intermediate time $t_2$. 

We denote the action of a control comb $\mathbf{Z}_{\hat{n}\hat{m}}$ on $\mathbf{T}_{\hat{n}}$ in terms of a a `braket' notation~\cite{extractingdynamicalquantumresources} as
\begin{gather}\label{eq:moreoptimal}
 \mathbf{T}'_{\hat{m}} = \llbracket \mathbf{T}_{\hat{n}} |\mathbf{Z}_{\hat{n} \hat{m}} \rrbracket  = \llbracket \mathbf{T}_{\hat{n}} |\mathbf{Z}_{\hat{n} \hat{n}} |\mathbf{I}_{\hat{n} \setminus \hat{m}} \rrbracket .
\end{gather} 
Mathematically this corresponds to contraction of the indices corresponding to times in $\hat{n}$ but not in $\hat{m}$ using the link product~\cite{theoreticalframeworkforquantumnetworks}, i.e., $\llbracket \mathbf{T}_{\hat{n}} |\mathbf{Z}_{\hat{n} \hat{m}} \rrbracket := d_s^{|\mathbf{T} \cap \mathbf{Z}|}\text{tr}_{\mathbf{T} \cap \mathbf{Z}}[\mathbf{T}_{\hat{n}}^{\mathrm{T}} \mathbf{Z}_{\hat{n} \hat{m}}]$, where $\text{tr}_{\mathbf{T} \cap \mathbf{Z}}$ is the trace over the spaces that both $\mathbf{T}_{\hat{n}}$ and $\mathbf{Z}_{\hat{n} \hat{m}} $ are defined on, $|\mathbf{T} \cap \mathbf{Z}|$ is the number of subsystems that are to be traced over, and $\bullet^\mathrm{T}$ denotes the transpose.\footnote{Here, the definition of the link product differs from the original one in Ref.~\cite{theoreticalframeworkforquantumnetworks} by a factor $d_s^{|\mathbf{T} \cap \mathbf{Z}|}$. This stems from the fact that the Choi matrices we use are normalised to unit trace.} See Fig.~\ref{fig:superprocess} for a graphical depiction. $\mathbf{Z}_{\hat{n} \hat{m}}$ and $\mathbf{Z}_{\hat{n} \hat{n}}$ are both superprocesses -- transformations on process tensors, in the same sense that superchannels are transformations of quantum channels~\cite{resourcetheoriesofmultitime}. Throughout, we employ the convention that $\mathbf{Z}_{\hat{n} \hat{m}}$ is the superprocess obtained from $\mathbf{Z}_{\hat{n} \hat{n}}$ via \textit{temporal coarse-graining}, i.e., by inserting identity maps into $\mathbf{Z}_{\hat{n} \hat{n}}$ at times $\hat{n} \setminus \hat{m}$. This operation is denoted by $\mathbf{Z}_{\hat{n} \hat{m}} = \llbracket \mathbf{Z}_{\hat{n} \hat{n}} |\mathbf{I}_{\hat{n} \setminus \hat{m}} \rrbracket = \text{tr}_{\hat{n}\setminus \hat{m}}[\mathbf{Z}_{\hat{n} \hat{n}} \mathbf{I}^\mathrm{T}_{\hat{n} \setminus \hat{m}}$], where $\text{tr}_{\hat{n}\setminus \hat{m}}$ is the trace over all Hilbert spaces that pertain to the times $\hat{n} \setminus \hat{m}$ and $\mathbf{I}_{\hat{n} \setminus \hat{m}}$ consists solely of Choi states of identity channels at times $\hat{n} \setminus \hat{m}$. That is, $\mathbf{I}_{\hat{n}\setminus\hat{m}} = \bigotimes_{t_k \in \hat{n} \setminus \hat{m}} \Psi^+_{t_k}$, where $\Psi^+_{t_k}$ is the maximally entangled state on $\mathcal{H}_{\text{out}_k} \otimes \mathcal{H}_{\text{in}_k}$. 

Indeed, \textit{any} superprocess $\mathbf{Z}_{\hat{n}\hat{m}}$ that changes the number of timesteps can always be realised from one at the finest level of temporal resolution (here: $\hat{n}$), followed by temporal coarse-graining~\cite{extractingdynamicalquantumresources}, justifying the second equality in Eq.~\eqref{eq:moreoptimal}. Consequently, $\mathbf{Z}_{\hat{n} \hat{n}}$ can be understood as an object that can act to the left -- yielding a transformed process $\mathbf{T}'_{\hat{n}} = \llbracket \mathbf{T}_{\hat{n}} |\mathbf{Z}_{\hat{n} \hat{n}}$ --  or to the right -- yielding a transformed control $\mathbf{Z}_{\hat{n}\hat{m}} =  \mathbf{Z}_{\hat{n} \hat{n}} |\mathbf{I}_{\hat{n} \setminus \hat{m}} \rrbracket$. This is the motivation for the `braket' notation we employ. Throughout, our results will primarily be derived from the properties of these objects at the abstract `braket'-level, reducing the need to explicitly compute link products in most cases.

\begin{figure*}[ht!] 
        \centering
\includegraphics[width=0.75\linewidth]{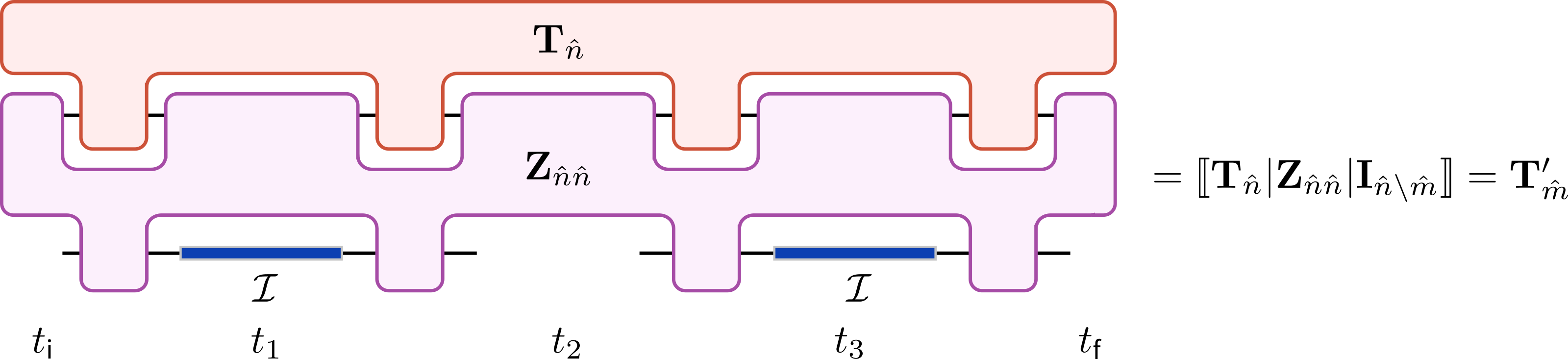}
\caption{\textbf{Action of a superprocess and coarse-graining.} A superprocess $\mathbf{Z}_{\hat{n}\hat{n}}$ maps a process on times $\hat{n}$ to another process on times $\hat{n}$. Temporal coarse-graining -- depicted by the identity channel $\mathcal{I}$ -- at times $\hat{n}\setminus \hat{m}$ reduces the temporal resolution from $\hat{n}$ to $\hat{m}$. The resulting process $\mathbf{T}'_{\hat{m}}$ is defined on the times $\hat{m}$ and given by $\llbracket \mathbf{T}_{\hat{n}}|\mathbf{Z}_{\hat{n}\hat{n}}|\mathbf{I}_{\hat{n}\setminus \hat{m}}\rrbracket = \llbracket \mathbf{T}_{\hat{n}}|\mathbf{Z}_{\hat{n}\hat{m}} \rrbracket$. Here, $\hat{n} = \{t_\text{i},t_1, t_2, t_3, t_\text{f}\}$ and $\hat{m}= \{t_\text{i},t_2, t_\text{f}\}$. }
\label{fig:superprocess}
\end{figure*}

\subsection{Restricted control combs} \label{sec:restrictedcombs}
In this work, we seek to quantify the utility of a quantum process under specified operational constraints, corresponding to limitations on the set $\mathsf{Z}$ of experimentally possible transformations $\mathbf{Z}_{\hat{n} \hat{m}}$. While the tools we present apply to any quantum control task, we focus on the problem of noise reduction and dynamical decoupling protocols, as in Ref.~\cite{extractingdynamicalquantumresources}. There, independent operations on the system $s$ at times $\hat{n}$ are employed to denoise its evolution. Consequently, our primary attention is this type of constraint, which corresponds to the \textit{independent quantum instruments} (IQI)~\cite{extractingdynamicalquantumresources} scenario. Specifically, IQI only allows for pre- and post-processing of each slot by means of independent quantum operations, potentially followed by coarse-graining, i.e., insertion of an identity map (see Fig.~\ref{fig:processtensorsuperprocess}). At each individual time, adaptive behaviour is allowed (e.g., performing measurements directly followed by conditional state preparations), but no memory is transmitted between individual times. Throughout, we only consider \textit{deterministic} control, where the action of each map is given by a quantum channel, i.e., a CPTP map.  

Structurally, the restriction to IQI superprocesses implies that the allowed superprocesses $\mathbf{Z}_{\hat{n}\hat{n}}$ at the finest level of temporal access $\hat{n}$ are themselves Markovian processes as in Eq.~\eqref{eq:Markovian_process}, i.e. all allowed $\mathbf{Z}_{\hat{n}\hat{n}} \in \mathsf{Z}_{\hat{n}\hat{n}}$ satisfy
\begin{equation}
    \mathbf{Z}_{\hat{n}\hat{n}} = \mathbf{W}_{t_{\text{f}}} \otimes \mathbf{V}_{t_{n}} \otimes \mathbf{W}_{t_{n}} \otimes \dots \otimes \mathbf{V}_{t_{1}} \otimes \mathbf{W}_{t_{1}} \otimes \mathbf{V}_{t_{\text{i}}},
\end{equation}
where $\mathbf{V}_{t_{k}} \in \mathcal{L}(\mathcal{H}_{\text{out}_k} \otimes \mathcal{H}_{\text{out}_k'})$ is the Choi state of a pre-processing operation $\mathcal{V}^s_{t_{k}}: \mathcal{L}(\mathcal{H}_{\text{out}_k'}) \rightarrow \mathcal{L}(\mathcal{H}_{\text{out}_k})$ at time $t_k$, and $\mathbf{W}_{t_{k+1}} \in \mathcal{L}(\mathcal{H}_{\text{in}_{k+1}'} \otimes \mathcal{H}_{\text{in}_{k+1}})$ is the Choi state of the corresponding post-processing operation $\mathcal{W}^s_{t_{k+1}}: \mathcal{L}(\mathcal{H}_{\text{in}_{k+1}}) \rightarrow \mathcal{L}(\mathcal{H}_{\text{in}'_{k+1}})$ (see Fig.~\ref{fig:processtensorsuperprocess} for a graphical depiction). As a result, we have $\llbracket\mathbf{T}_{\hat{n}}|\mathbf{Z}_{\hat{n}\hat{n}} =: \mathbf{T}_{\hat{n}}' \in \mathcal{L}(\mathcal{H}_{\text{in}'_\text{f}} \otimes \mathcal{H}_{\text{out}'_n} \otimes \mathcal{H}_{\text{in}'_n} \otimes \cdots \mathcal{H}_{\text{in}'_1} \otimes \mathcal{H}_{\text{out}'_\text{i}})$. We emphasise that all superprocesses $\mathbf{Z}_{\hat{n}\hat{m}} \in \mathsf{Z}_{\hat{n}\hat{m}}$ mapping the comb $\mathbf{T}_{\hat{n}}$ to one with fewer intermediate times $\hat{m}$ can be expressed as $\mathbf{Z}_{\hat{n}\hat{m}} = \mathbf{Z}_{\hat{n}\hat{n}}|\mathbf{I}_{\hat{n}\setminus \hat{m}}\rrbracket$ for some $\mathbf{Z}_{\hat{n}\hat{n}} \in \mathsf{Z}_{\hat{n}}$~\cite{ extractingdynamicalquantumresources}, i.e., they can be obtained from $\mathsf{Z}_{\hat{n}\hat{n}}$ via temporal coarse-graining.

Correlations in the process tensor $\mathbf{T}_{\hat{n}}$ constitute a resource that can be exploited via superprocesses $\mathbf{Z}_{\hat{n}\hat{n}}$ for better performance at a given task (like, for example, denoising). Consequently, the set of \textit{resourceless/free processes} is the set of those processes that are \textit{fully uncorrelated}, i.e., of the form
\begin{align} \label{eq:fullyuncorrelated}
\notag
    \mathbf{T}^{\mathscr{U}}_{\hat{n}} &= \mathbf{T}_{t_n} \otimes \mathbf{T}_{t_{n-1}} \otimes \cdots \otimes \mathbf{T}_{t_1} \otimes \mathbf{T}_{t_\text{i}}\\
    \notag
    &= (\rho_{\text{in}_\text{f}} \otimes \mathbbm{1}_{\text{out}_n}/d_{\text{out}_n}) \otimes (\rho_{\text{in}_n} \otimes \mathbbm{1}_{\text{out}_{n-1}}/d_{\text{out}_{n-1}}) \otimes \cdots \\
    &\phantom{=}\cdots\otimes (\rho_{\text{in}_2} \otimes \mathbbm{1}_{\text{out}_1}/d_{\text{out}_1}) \otimes (\rho_{\text{in}_1} \otimes \mathbbm{1}_{\text{out}_\text{i}}/d_{\text{out}_\text{i}}) .
\end{align}
Since these processes do not display any correlations, they cannot preserve any information in time. IQI superprocesses $\mathbf{Z}_{\hat{n}\hat{n}} \in \mathsf{Z}_{\hat{n}\hat{n}}$ map such resourceless processes onto resourceless processes~\cite{resourcetheoriesofmultitime}, i.e., IQI processes cannot create temporal correlations. Consequently, the presence of correlations within $\mathbf{T}_{\hat{m}}$  -- or, equivalently, its ability to preserve information in time -- is deemed a resource within the framework of IQI. In what follows, we will analyse how this resource can be employed by means of IQI superprocesses for the task of noise reduction.

We emphasise that the allowed experimental control within IQI encompasses traditional (Markovian) error correction methods as well as all fully open-loop control, including dynamical decoupling (DD).

\begin{figure*}[ht!] 
        \centering
        \begin{tikzpicture}[scale=0.4]
\draw[ mediumgrey, thick,dashed] (0+0.0,-3) -- (0+0.0,-7);
\draw[ mediumgrey, thick,dashed] (2+0.0,-3) -- (2+0.0,-7);
\draw[ mediumgrey, thick,dashed] (0+8.0,-3) -- (0+8.0,-7);
\draw[ mediumgrey, thick,dashed] (2+8.0,-3) -- (2+8.0,-7);
\draw[ mediumgrey, thick,dashed] (0+16.0,-3) -- (0+16.0,-7);
\draw[ mediumgrey, thick,dashed] (2+16.0,-3) -- (2+16.0,-7);
\draw[ mediumgrey, thick,dashed] (0+24.0,-3) -- (0+24.0,-7);
\draw[ mediumgrey, thick,dashed] (2+24.0,-3) -- (2+24.0,-7);

\draw[myred,fill=myredfill, thick,solid,rounded corners=4] (-2,3-4) -- (-2,4-4) -- (28,4-4) -- (28,2+0.1-4) -- (26-0.1,2+0.1-4) -- (26-0.1,0+0.1-4) -- (24+0.1,0+0.1-4) -- (24+0.1,2+0.1-4)-- (18-0.1,2+0.1-4) -- (18-0.1,0+0.1-4) -- (16+0.1,0+0.1-4) -- (16+0.1,2+0.1-4) -- (10-0.1,2+0.1-4) -- (10-0.1,0+0.1-4) -- (8+0.1,0+0.1-4) -- (8+0.1,2+0.1-4) -- (2-0.1,2+0.1-4) -- (2-0.1,0+0.1-4) -- (0+0.1,0+0.1-4) -- (0+0.1,2+0.1-4) -- (-2,2+0.1-4) -- (-2,3-4)   ;

\draw[black, very thick,solid] (-1,3-4) -- (27,3-4);
\draw[black, very thick,solid] (26.8,2.7-4) -- (27.2,3.3-4);
\draw[black, very thick,solid] (-1.5,1-4) -- (27.5,1-4);

\draw[myred,fill=myredfill,very thick,solid,rounded corners=2] (-1.7,3.7-4) rectangle (-0.3,2.3-4);

\draw[myred,fill=myredfill,very thick,solid,rounded corners=2] (0.3,3.7-4) rectangle (1.7,0.3-4);
\draw[myred,fill=myredfill,very thick,solid,rounded corners=2] (8.3,3.7-4) rectangle (9.7,0.3-4);
\draw[myred,fill=myredfill,very thick,solid,rounded corners=2] (16.3,3.7-4) rectangle (17.7,0.3-4);
\draw[myred,fill=myredfill,very thick,solid,rounded corners=2] (24.3,3.7-4) rectangle (25.7,0.3-4);

\draw[white,fill=white,ultra thick,solid,rounded corners=2] (-1.7,1.7-4) rectangle (-0.3,-1.7-4);
\draw[white,fill=white,ultra thick,solid,rounded corners=2] (4.3-2,1.7-4) rectangle (5.7+2,-1.7-4);
\draw[white,fill=white,ultra thick,solid,rounded corners=2] (12.3-2,1.7-4) rectangle (13.7+2,-1.7-4);
\draw[white,fill=white,ultra thick,solid,rounded corners=2] (20.3-2,1.7-4) rectangle (21.7+2,-1.7-4);
\draw[white,fill=white,ultra thick,solid,rounded corners=2] (26.3,1.7-4) rectangle (27.7,-1.7-4);

\draw[] (-1,3-4) node {\small $\rho^{\text{e}}_{t_{\text{i}}}$};
\draw[] (1,2-4) node[rotate=90] {\small $\mathcal{T}_{{t_1:t_{\text{i}}}}$};
\draw[] (9,2-4) node[rotate=90] {\small $\mathcal{T}_{t_2:t_1}$};
\draw[] (17,2-4) node[rotate=90] {\small $\mathcal{T}_{t_3:t_2}$};
\draw[] (25,2-4) node[rotate=90] {\small $\mathcal{T}_{t_{\text{f}}:t_3}$};

\draw[] (-1,1-4) node { $t_{\text{i}}$};
\draw[] (5,1-4) node { $t_1$};
\draw[] (13,1-4) node { $t_2$};
\draw[] (21,1-4) node { $t_3$};
\draw[] (27,1-4) node { $t_{\text{f}}$};

\draw[] (13,5-4-0.25) node {\large Noise process $\mathbf{T}_{\hat{n}}$};


\draw[black, very thick,solid] (-2.1,-7) -- (0.1,-7);
\draw[black, very thick,solid] (-2.1+4,-7) -- (0.1+4,-7);
\draw[black, very thick,solid] (-2.1+8,-7) -- (0.1+8,-7);
\draw[black, very thick,solid] (-2.1+12,-7) -- (0.1+12,-7);
\draw[black, very thick,solid] (-2.1+16,-7) -- (0.1+16,-7);
\draw[black, very thick,solid] (-2.1+20,-7) -- (0.1+20,-7);
\draw[black, very thick,solid] (-2.1+24,-7) -- (0.1+24,-7);
\draw[black, very thick,solid] (-2.1+28,-7) -- (0.1+28,-7);

\draw[mypurple,fill=mypurplefill,very thick,solid,rounded corners=2] (-1.7,1.7-8) rectangle (-0.3,-1.7-8+2);
\draw[mypurple,fill=mypurplefill,very thick,solid,rounded corners=2] (4.3-2,1.7-8) rectangle (5.7-2,-1.7-8+2);
\draw[mypurple,fill=mypurplefill,very thick,solid,rounded corners=2] (4.3+2,1.7-8) rectangle (5.7+2,-1.7-8+2);
\draw[mypurple,fill=mypurplefill,very thick,solid,rounded corners=2] (12.3-2,1.7-8) rectangle (13.7-2,-1.7-8+2);
\draw[mypurple,fill=mypurplefill,very thick,solid,rounded corners=2] (12.3+2,1.7-8) rectangle (13.7+2,-1.7-8+2);
\draw[mypurple,fill=mypurplefill,very thick,solid,rounded corners=2] (18.3,1.7-8) rectangle (19.7,-1.7-8+2);
\draw[mypurple,fill=mypurplefill,very thick,solid,rounded corners=2] (20.3+2,1.7-8) rectangle (21.7+2,-1.7-8+2);
\draw[mypurple,fill=mypurplefill,very thick,solid,rounded corners=2] (26.3,1.7-8) rectangle (27.7,-1.7-8+2);

\draw[] (-1,0-7) node[rotate=0] {\small $\mathcal{V}_{t_{\text{i}}}$};
\draw[] (5-2,0-7) node[rotate=0] {\small $\mathcal{W}_{t_{1}}$};
\draw[] (5+2,0-7) node[rotate=0] {\small $\mathcal{V}_{t_{1}}$};
\draw[] (13-2,0-7) node[rotate=0] {\small $\mathcal{W}_{t_{2}}$};
\draw[] (13+2,0-7) node[rotate=0] {\small $\mathcal{V}_{t_{2}}$};
\draw[] (21-2,0-7) node[rotate=0] {\small $\mathcal{W}_{t_{3}}$};
\draw[] (21+2,0-7) node[rotate=0] {\small $\mathcal{V}_{t_{3}}$};
\draw[] (27,0-7) node[rotate=0] {\small $\mathcal{W}_{t_{\text{f}}}$};

\draw[] (13,-5-0.25) node {Memoryless superprocess $\mathbf{Z}_{\hat{n}\hat{n}}$ acts on $\mathbf{T}_{\hat{n}}$};


\draw[black, very thick,solid] (-1,3-18+2+3) -- (27,3-18+2+3);
\draw[black, very thick,solid] (26.8,2.7-18+2+3) -- (27.2,3.3-18+2+3);
\draw[black, very thick,solid] (-2,1-18+2+3) -- (28,1-18+2+3);

\draw[myred,fill=myredfill, thick,solid,rounded corners=4] (-2,3-18+2+3) -- (-2,4-18+2+3) -- (28,4-18+2+3) -- (28,2+0.1-18+2+3) -- (26-0.1,2+0.1-18+2+3) -- (26-0.1,0+0.1-18+2+3) -- (24+0.1,0+0.1-18+2+3) -- (24+0.1,2+0.1-18+2+3)-- (18-0.1,2+0.1-18+2+3) -- (18-0.1,0+0.1-18+2+3) -- (16+0.1,0+0.1-18+2+3) -- (16+0.1,2+0.1-18+2+3) -- (10-0.1,2+0.1-18+2+3) -- (10-0.1,0+0.1-18+2+3) -- (8+0.1,0+0.1-18+2+3) -- (8+0.1,2+0.1-18+2+3) -- (2-0.1,2+0.1-18+2+3) -- (2-0.1,0+0.1-18+2+3) -- (0+0.1,0+0.1-18+2+3) -- (0+0.1,2+0.1-18+2+3) -- (-2,2+0.1-18+2+3) -- (-2,3-18+2+3)   ;

\draw[white,fill=white,ultra thick,solid,rounded corners=2] (4.3,1.7-18+2+3) rectangle (5.7,-1.7-18+2+3);
\draw[white,fill=white,ultra thick,solid,rounded corners=2] (12.3,1.7-18+2+3) rectangle (13.7,-1.7-18+2+3);
\draw[white,fill=white,ultra thick,solid,rounded corners=2] (20.3,1.7-18+2+3) rectangle (21.7,-1.7-18+2+3);

\draw[] (13,3-18+2+3) node {\large Transformed process $\mathbf{T}'_{\hat{n}}$};


\draw[ mediumgrey, thick,dashed] (4+0.3,-15+3) -- (4+0.3,-17+3-1);
\draw[ mediumgrey, thick,dashed] (6-0.3,-15+3) -- (6-0.3,-17+3-1);
\draw[ mediumgrey, thick,dashed] (4+0.3+8,-15+3) -- (4+0.3+8,-17+3-1);
\draw[ mediumgrey, thick,dashed] (6-0.3+8,-15+3) -- (6-0.3+8,-17+3-1);
\draw[ mediumgrey, thick,dashed] (4+0.3+16,-15+3) -- (4+0.3+16,-17+3-1);
\draw[ mediumgrey, thick,dashed] (6-0.3+16,-15+3) -- (6-0.3+16,-17+3-1);

\draw[myblue,line width=3pt] (4.3,1-26+8+3-1) -- (5.7,1-26+8+3-1);

\draw[myblue,line width=3pt] (12.3,1-26+8+3-1) -- (13.7,1-26+8+3-1);

\draw[myblue,line width=3pt] (20.3,1-26+8+3-1) -- (21.7,1-26+8+3-1);

\draw[] (5,-18.5+3-0.25) node {\small $\mathcal{I}$};
\draw[] (13,-18.5+3-0.25) node {\small $\mathcal{I}$};
\draw[] (21,-18.5+3-0.25) node {\small $\mathcal{I}$};

\draw[] (13,2-26+8+3-1) node {Coarse-graining $\mathbf{I}_{\hat{n} \setminus \hat{m}}$ acts on $\mathbf{T}'_{\hat{n}}$};

\draw[line width=1pt,mediumgrey, decoration={brace,raise=0pt,amplitude=8pt},decorate]
  (-2.3,-8) --  (-2.3,0);

  \draw[] (-2.3-1.5,-4) node[rotate=90] {\small $\mathbf{T}'_{\hat{n}} =\llbracket \mathbf{T}_{\hat{n}} |\mathbf{Z}_{\hat{n} \hat{n}} \rrbracket$};

  \draw[line width=1pt,mediumgrey, decoration={brace,raise=0pt,amplitude=8pt},decorate]
  (-2.3,-15) --  (-2.3,-9);
  \draw[] (-2.3-1.5,-12) node[rotate=90] {\small $\mathbf{T}'_{\emptyset} = \llbracket \mathbf{T}'_{\hat{n}} |\mathbf{I}_{\hat{n} \setminus \emptyset} \rrbracket$};

\end{tikzpicture}

\caption{\textbf{Noise reduction in IQI.} The process tensor $\mathbf{T}_{\hat{n}}$ is accessible at intermediate times $\hat{n}$, plus initial and final times $t_{\text{i}},t_{\text{f}}$. The superprocess $\mathbf{Z}_{\hat{n}\hat{n}}$ corresponds to experimental control, and transforms the process tensor $\mathbf{T}_{\hat{n}}$ on $\hat{n}$ times to another process $\mathbf{T}_{\hat{n}}'$ on the same set of times. In the operational scenario of IQI experimental control is memoryless, implying that $\mathbf{Z}_{\hat{n}\hat{n}}$ is realised by predetermined sequences of pre- and post-processing channels $\mathcal{V}_{t_{k}}: \mathcal{L}(\mathcal{H}_{\text{out}'_k}) \rightarrow \mathcal{L}(\mathcal{H}_{\text{out}_k})$ and $\mathcal{W}_{t_{k+1}}: \mathcal{L}(\mathcal{H}_{\text{in}_{k+1}}) \rightarrow \mathcal{L}(\mathcal{H}_{\text{in}'_{k+1}})$. Note that, in principle, $\mathcal{V}_{t_{k}}$ and $\mathcal{W}_{t_{k+1}}$ can change the dimension of the involved spaces, i.e., we can have $\mathcal{H}_{\text{in}_k} \ncong \mathcal{H}_{\text{in}'_k}$ and $\mathcal{H}_{\text{out}_k} \ncong \mathcal{H}_{\text{out}_k'}$. Temporal coarse-graining $\mathbf{I}_{\hat{n} \setminus \hat{m}}$ removes times from $\hat{n}$ that are not in $\hat{m}$. The resource quantifiers $\overline{I}_{\hat{m}},\overline{M}_{\hat{m}},\overline{N}_{\hat{m}}$ are concerned with the temporal correlations of the process at the coarse-grained level $\hat{m}$ of access. As depicted here, $\hat{m}$ is the empty set $\emptyset$, meaning that the resultant process after coarse-graining has no intermediate times. The relevant quantifier $\overline{I}_{\emptyset}=\overline{M}_{\emptyset}$ then becomes the maximum mutual information of the resultant channel that one can obtain under superprocesses $\mathbf{Z}_{\hat{n}\hat{n}}$ from IQI. If this mutual information of the resulting channel is larger than that of the one one would have obtained from performing \textit{no} operations (i.e., only temporal coarse-graining), then the underlying resources within $\mathbf{T}_{\hat{n}}$ have successfully be employed for denoising.} \label{fig:processtensorsuperprocess}
\end{figure*}
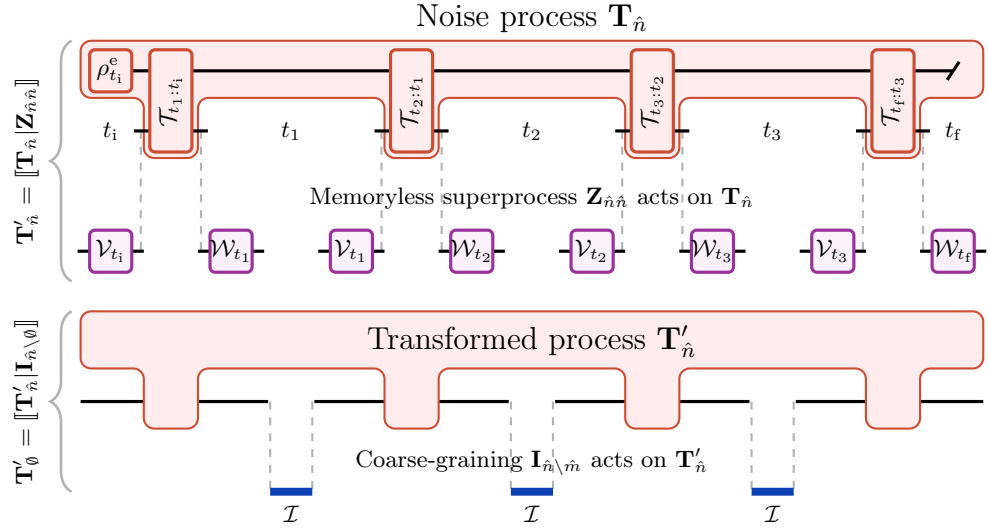

For the remainder of the manuscript, we will focus on DD scenarios, primarily because quantum error correction only works in the many-copy regime, while DD can reduce noise on a single copy of a quantum system. Ref.~\cite{extractingdynamicalquantumresources} interpreted this as \textit{temporal} resource distillation, as opposed to the spatial resource distillation of quantum error correction. We aim to formally quantify this effect.

\section{Resource Quantifiers} \label{sec:monotones}
Temporal correlations in a multi-time process $\mathbf{T}_{\hat{n}}$ constitute a resource for many practical tasks. To quantify them, we consider two distinct marginal processes of $\mathbf{T}_{\hat{n}}$,
\begin{gather} \label{eq:marginals}
     \mathbf{T}^{\text{Mkv}}_{\hat{n}} \! := \! \bigotimes_{j\in \text{J}}  {\rm tr}_{\bar{j}} \{ \mathbf{T}_{\hat{n}}\}
     \ \ \ \mbox{and} \ \ \
     \mathbf{T}^{\text{marg}}_{\hat{n}} \! := \! \bigotimes_{k\in K}  {\rm tr}_{\bar{k}} \{ \mathbf{T}_{\hat{n}}\}.
\end{gather}
The index set $\text{J} = \{\text{out}_\text{i}\text{in}_1, \text{out}_\text{1}\text{in}_2, \text{out}_\text{2}\text{in}_3, \dots, \text{out}_\text{n}\text{in}_\text{f}\}$ with $|\text{J}| = n+1$ enumerates the constituent channels of the process, and the set $\text{K} = \{\text{out}_\text{i}, \text{in}_1, \text{out}_1, \dots, \text{in}_n, \text{out}_n, \text{in}_\text{f}\}$ with $|\text{K}| = 2(n+1)$ splits this further into the marginal inputs and outputs of those channels. The overline on an index signifies its complement. As such, $\mathbf{T}^{\text{Mkv}}_{\hat{n}}$ only contains the Markovian temporal correlations of $\mathbf{T}_{\hat{n}}$, while  $\mathbf{T}^{\text{marg}}_{\hat{n}}$ contains no temporal correlations whatsoever. These are the nearest Markovian and uncorrelated processes to $\mathbf{T}_{\hat{n}}$, respectively, as measured by quantum relative entropy $S(x\|y) := \mbox{tr}\{x \log(x) - x \log(y)\}$.\footnote{Here, all process tensors are normalised to unit trace, such that entropic measures are well-defined} Evidently, if $\mathbf{T}_{\hat{n}}$ is Markovian, i.e., it is of the form of Eq.~\eqref{eq:Markovian_process}, then $\mathbf{T}_{\hat{n}} = \mathbf{T}^{\text{Mkv}}_{\hat{n}}$, but generally $\mathbf{T}_{\hat{n}} \neq \mathbf{T}^{\text{marg}}_{\hat{n}}$, while non-Markovian processes satisfy $\mathbf{T}_{\hat{n}} \neq \mathbf{T}^{\text{Mkv}}_{\hat{n}}$ and $\mathbf{T}_{\hat{n}} \neq \mathbf{T}^{\text{marg}}_{\hat{n}}$. Given the above definitions, we can define the total information $I$, the Markovian information $M$, and the non-Markovianity $N$ as
\begin{gather}
\begin{split}
    \label{eq:Imonotoneresolution}
    I(\mathbf{T}_{\hat{n}}) \! &:= \! S\left(\mathbf{T}_{\hat{n}} \| \mathbf{T}^{\text{marg}}_{\hat{n}} \right), \\
    M(\mathbf{T}_{\hat{n}}) &:=S\left( \mathbf{T}^{\text{Mkv}}_{\hat{n}} \| \mathbf{T}^{\text{marg}}_{\hat{n}} \right), \\
    \mbox{and} \ \ N(\mathbf{T}_{\hat{n}}) \! &:= \! S\left(\mathbf{T}_{\hat{n}} \| \mathbf{T}^{\text{Mkv}}_{\hat{n}} \right)  .
\end{split}
\end{gather}

$I,M,N$ are not monotonic under the constraints of IQI, in the sense that there exist superprocesses $\mathbf{Z}_{\hat{n}\hat{m}}$ in IQI, such that, e.g., $I(\llbracket \mathbf{T}_{\hat{n}}|\mathbf{Z}_{\hat{n}\hat{m}}) \geq I(\mathbf{T}_{\hat{n}})$~\cite{zambon2024processtensordistinguishabilitymeasures}. As a consequence, these functions cannot be interpreted as resource quantifiers on their own, since they can increase under the superprocesses that IQI allows for. 

A concrete example of this failure of monotonicity for a $1$-slot comb was provided in Ref.~\cite{zambon2024processtensordistinguishabilitymeasures}: Consider a qubit system $s$ with a two qubit environment $\mathcal{H}_e= \mathcal{H}_{e_1}\otimes \mathcal{H}_{e_2}$ that is in the initial state $\rho_{t_\text{i}}^{e} = \ketbra{0^{e_1}}{0^{e_1}} \otimes \mathbbm{1}^{e_2}/2$. The $|\hat{n}|=1$ intermediate time process $\mathbf{T}_{\hat{n}}$ stems from two intermediate system-environment dynamics between $t_\text{i}$ and $t_1$ and between $t_1$ and $t_\text{f}$, respectively. The former corresponds to a $se_1$ SWAP gate. The second system-environment dynamics consists of a concatenation of a dephasing gate on $s$, followed by an $s$-controlled $e_1 e_2$ SWAP gate, and finally an $s e_1$ SWAP gate. 

The overall behaviour of $\mathbf{T}_{\hat{n}}$ is that the first output of the system at time $t_1$ is $\ketbra{0}{0}$ regardless of the input at time $t_\text{i}$, and the output of the process at time $t_\text{f}$ is $p \rho^s_1 + (1-p)\mathbbm{1}/2$, where $p = \braket{0|\rho^s_2|0}$ and $\rho^s_2$ is the state of the system $s$ entering the process at $t_{1}$. It can be directly verified that, overall, this process satisfies $I(\mathbf{T}_{\hat{n}})=1$. However, it is easy to see that inserting an identity channel at time $t_1$ yields a coarse-grained process $\mathbf{T}_{\emptyset}$ that corresponds to an identity channel from $t_\text{i}$ to $t_\text{f}$, resulting in $I(\mathbf{T}_{\emptyset})=2$. This value is higher than the one before the coarse-graining, even though the performed operation on $\mathbf{T}_{\hat{n}}$ that lies within IQI. From this example, we can see that the Choi state of a process tensor alone does not capture the potential correlations it can exhibit when acted upon (e.g., by coarse-graining it). Due to the non-monotonicity of relative entropy under such operations, the correlations contained in the resulting process tensor can be greater than those of the Choi state itself. To remedy this issue, and to construct a well-behaved resource quantifier based on the relative entropy, \textit{all} possible choices of control must be considered to ensure monotonicity.

In particular, the \textit{optimal} value of any function -- such as $I,M,N$ -- that can be obtained subject to the allowed transformations (like, e.g., those in IQI)  must be non-increasing under said allowed transformations as well as temporal coarse-graining. This is a direct consequence of the fact that all transformations of the form $\mathbf{Z}_{\hat{n}\hat{m}}$ can be expressed as $\mathbf{Z}_{\hat{m}\hat{n}} = \mathbf{Z}_{\hat{m}\hat{n}}|\mathbf{I}_{\hat{n}\setminus \hat{m}} \rrbracket$, i.e., they stem from a control comb $\mathbf{Z}_{\hat{n}\hat{n}}$ followed by temporal coarse-graining at times $\hat{n} \setminus \hat{m}$. On the other hand, more processes can be `reached' by first applying $\mathbf{Z}_{\hat{n}\hat{n}} \in \mathsf{Z}_{\hat{n}\hat{n}}$ on $\mathbf{T}_{\hat{n}}$ and then temporal coarse-graining at times $\hat{n}\setminus \hat{m}$, than by first applying temporal coarse-graining and then applying a `lower resolution' superprocess $\mathbf{Z}_{\hat{m}\hat{m}} \in\mathsf{Z}_{\hat{m}\hat{m}}$. Put differently, there are processes $\mathbf{T}_{\hat{n}}$ and superprocesses $\mathsf{Z}_{\hat{n}\hat{n}} \in \mathsf{Z}_{\hat{n}}$, such that $\llbracket \mathbf{T}_{\hat{n}}|\mathbf{Z}_{\hat{n}\hat{n}}|\mathbf{I}_{\hat{n}\setminus \hat{m}}\rrbracket$ can \textit{not} be expressed as $\big\llbracket \llbracket\mathbf{T}_{\hat{n}}|\mathbf{I}_{\hat{n}\setminus \hat{m}}\rrbracket\big\vert \mathbf{Z}_{\hat{m}\hat{m}}\rrbracket \big\rrbracket$ for any $\mathbf{Z}_{\hat{m}\hat{m}} \in \mathsf{Z}_{\hat{m}\hat{m}}$. 

This property, that having less temporal resolution can never increase the range of potential control, is referred to as the \textit{irreversibility of temporal coarse-graining}~\cite{extractingdynamicalquantumresources}. With this, we can use $I, M$, and $N$ to define three families of resource quantifiers that satisfy the desired monotonicity property:
\begin{theorem}[Irreversibility Resource Quantifiers]
\label{thm::irrev}
Let $\hat{m} \subseteq \hat{n}$. Then,
\begin{gather}
\begin{split}
    \overline{I}_{\hat{m}}(\mathbf{T}_{\hat{n}}) \! &:= \! \sup_{\mathbf{Z}_{\hat{n}\hat{n}} \in \mathsf{Z}_{\hat{n}\hat{n}}} I\Big( \llbracket \mathbf{T}_{\hat{n}} | \mathbf{Z}_{\hat{n}\hat{n}} | \mathbf{I}_{{\hat{n}}\setminus{\hat{m}}} \rrbracket  \Big), \\
    \overline{M}_{\hat{m}}(\mathbf{T}_{\hat{n}}) \! &:= \! \sup_{\mathbf{Z}_{\hat{n}\hat{n}} \in \mathsf{Z}_{\hat{n}\hat{n}}} M\Big( \llbracket \mathbf{T}_{\hat{n}} | \mathbf{Z}_{\hat{n}\hat{n}} | \mathbf{I}_{{\hat{n}}\setminus{\hat{m}}} \rrbracket  \Big), \\
    \mbox{\text{\emph{and}}} \ \ \overline{N}_{\hat{m}}(\mathbf{T}_{\hat{n}}) \! &:= \! \sup_{\mathbf{Z}_{\hat{n}\hat{n}} \in \mathsf{Z}_{\hat{n}\hat{n}}} N\Big( \llbracket \mathbf{T}_{\hat{n}} | \mathbf{Z}_{\hat{n}\hat{n}} | \mathbf{I}_{{\hat{n}}\setminus{\hat{m}}} \rrbracket  \Big),
\end{split}
\end{gather}
are monotonic in IQI.
\end{theorem}
Thm.~\ref{thm::irrev} holds true for \textit{any} set of free (i.e. allowed) operations $\mathsf{Z}_{\hat{n}\hat{n}}$, in particular for the case where $\mathsf{Z}_{\hat{n}\hat{n}}$ is the set of uncorrelated sequences of channels and temporal coarse-graining, which corresponds to IQI, the resource theory of temporal resolution considered in Ref.~\cite{extractingdynamicalquantumresources}. Hence, Thm.~\ref{thm::irrev} guarantees that these quantifiers cannot be increased by any actions an experimenter operating within IQI can perform, implying that they meaningfully quantify the amenability of a process to denoising under those constraints.
\begin{proof}
We first consider the case $\hat{m} = \hat{n}$, i.e., the case without temporal coarse-graining. Let $\mathbf{Z}^{*}_{\hat{n}\hat{n}}$ be the optimal superprocess in the definition of $\overline{I}_{\hat{n}}(\mathbf{T}_{\hat{n}})$. Evidently, the composition of $\mathbf{Z}^{*}_{\hat{n}\hat{n}}$ with some other arbitrary $\mathbf{Z}_{\hat{n}\hat{n}} \in \mathsf{Z}_{\hat{n}\hat{n}}$ is by definition not more optimal than $\mathbf{Z}^{*}_{\hat{n}\hat{n}}$, implying the monotonicity of $\overline{I}_{\hat{m}}(\mathbf{T}_{\hat{n}})$ under transformations in $\mathsf{Z}_{\hat{n}\hat{n}}$ in the case where the number of steps is unchanged. Monotonicity of $\overline{M}_{\hat{m}}(\mathbf{T}_{\hat{n}})$ and $\overline{N}_{\hat{n}}(\mathbf{T}_{\hat{n}})$ follows by the same argument. 

The case $\hat{m} \subset \hat{n}$ follows from the irreversibility of temporal coarse-graining~\cite{extractingdynamicalquantumresources}: We note that the two sets 
\begin{align}
   &\mathsf{R}_{\hat{n}} = \big\{ \llbracket \mathbf{T}_{\hat{n}} | \mathbf{Z}_{\hat{n}\hat{n}} | \mathbf{I}_{{\hat{n}}\setminus{\hat{m}}} \rrbracket  , \mathbf{Z}_{\hat{n}\hat{n}} \in \mathsf{Z}_{\hat{n}\hat{n}} \big\}\\
\text{and} \quad &\mathsf{R}_{\hat{m}} = \left\{  \Big\llbracket \llbracket \mathbf{T}_{\hat{n}}  | \mathbf{I}_{{\hat{n}} \setminus{\hat{m}}} \rrbracket \Big\vert  \mathbf{Z}_{\hat{m}\hat{m}} \Big\rrbracket  , \mathbf{Z}_{\hat{m}\hat{m}} \in \mathsf{Z}_{\hat{m}\hat{m}} \right\}
\end{align}
satisfy $\mathsf{R}_{\hat{m}} \subseteq \mathsf{R}_{\hat{n}}$. As mentioned, this implies that the set of processes on times $\hat{m}$ that can be reached by first applying control operations and then coarse-graining is generally larger than the set that can be reached by first coarse-graining and then performing control operations on fewer steps. This holds true under the mild assumption that `doing nothing'/`asserting no control' at the times $\hat{n}\setminus{m}$ is possible~\cite{extractingdynamicalquantumresources}, which is satisfied in all conceivable control scenarios.

The inclusion $\mathsf{R}_{\hat{m}} \subseteq \mathsf{R}_{\hat{n}}$ implies that regardless of the function $X: \mathsf{T} \rightarrow \mathbbm{R}_{\geq0}$ we choose, it holds for all $\hat{m} \subseteq \hat{n}$ that
\begin{equation} \label{eq:monotoneirreversibility}
\begin{aligned}
   & \sup_{\mathbf{Z}_{\hat{n}\hat{n}} \in \mathsf{Z}_{\hat{n}\hat{n}}} X\Big( \llbracket \mathbf{T}_{\hat{n}} | \mathbf{Z}_{\hat{n}\hat{n}} | \mathbf{I}_{{\hat{n}}\setminus{\hat{m}}} \rrbracket  \Big) \\
   \geq & \sup_{\mathbf{Z}_{\hat{m}\hat{m}} \in \mathsf{Z}_{\hat{m}\hat{m}}} X\big( \Big\llbracket \llbracket \mathbf{T}_{\hat{n}}  | \mathbf{I}_{{\hat{n}} \setminus{\hat{m}}} \rrbracket \Big\vert  \mathbf{Z}_{\hat{m}\hat{m}} \Big\rrbracket \big).
\end{aligned}
\end{equation}
This is monotonicity under temporal coarse-graining. Finally, combining monotonicity under allowed superprocesses and under temporal coarse-graining yields monotonicity under all possible control in IQI. Since the respective properties of IQI were not used in this derivation, it applies to \textit{all} setups where `doing nothing' at times $\hat{n} \setminus \hat{m}$ is possible.
\end{proof}

These quantifiers have the physical interpretation of being the highest value one can obtain for $I,M,N$ at temporal resolution $\hat{m}$, given a process $\mathbf{T}_{\hat{n}}$ with temporal resolution $\hat{n}$. In particular, when $\hat{m}=\emptyset$, i.e. coarse-graining to a channel, $\overline{I}_{\emptyset}(\mathbf{T}_{\hat{n}})$ is the highest achievable value for the mutual information of that channel under the allowed control. Consequently, it possesses a direct interpretation in terms of the `best' (i.e., most noiseless) channel one can obtain from $\mathbf{T}_{\hat{n}}$ by means of control operations that lie in IQI.

It is easy to see that $\overline{I}_{\emptyset}$ in particular quantifies the amenability of a process to denoising under the constraints of IQI. If the experimenter does not attempt denoising, they will obtain the channel $\mathbf{T}_{\emptyset}:=\llbracket \mathbf{T}_{\hat{n}} | \mathbf{I}_{\hat{n} \setminus \emptyset} \rrbracket$. Then, the corresponding monotone is $\overline{I}_{\emptyset}(\mathbf{T}_{\emptyset})$. The goal of denoising is to find the optimal $\mathbf{Z}^*_{\hat{n}\hat{n}} \in \mathsf{Z}$ such that the transformed process $\mathbf{T}'_{\emptyset}:=\llbracket \mathbf{T}_{\hat{n}} | \mathbf{Z}^*_{\hat{n}\hat{n}} | \mathbf{I}_{\hat{n} \setminus \emptyset} \rrbracket$ satisfies $\overline{I}_{\emptyset}(\mathbf{T}_{\hat{n}})= \overline{I}_{\emptyset}(\mathbf{T}'_{\emptyset}) \geq   \overline{I}_{\emptyset}(\mathbf{T}_{\emptyset})$. The maximum possible gap $\overline{I}_{\emptyset}(\mathbf{T}'_{\emptyset}) - \overline{I}_{\emptyset}(\mathbf{T}_{\emptyset})$ then corresponds to the amenability of $\mathbf{T}_{\hat{n}}$ to denoising.

While $\overline{I}_{\hat{m}},\overline{M}_{\hat{m}},\overline{N}_{\hat{m}}$ are not as straightforward to compute as $I,M,N$, approximations of $\overline{I}_{\emptyset}$ -- or, more generally, $\overline{I}_{\hat{m}}$, $\overline{M}_{\hat{m}}$, and $\overline{N}_{\hat{m}}$ -- can still be found numerically, for example via the multitime optimal dynamical decoupling (MODD) algorithm of Ref.~\cite{extractingdynamicalquantumresources} (see Sec.~\ref{sec:numerics} for further discussion).

\section{Subadditivity of Resource Quantifiers} \label{sec:subadditivity}
The Markovian and non-Markovian correlations in process tensors satisfy intuitive monogamy relations~\cite{capela_monogamy_2020, Zambon_2024}. For example, a process that displays a large amount of Markovian correlations (i.e., between adjacent points in time) cannot simultaneously display a large amount of correlations between non-adjacent points in time, and vice versa. Given this, one should expect that the values of the monotones $\overline{I}_{\hat{m}},\overline{M}_{\hat{m}},\overline{N}_{\hat{m}}$ are not mutually independent. For the non-monotonic quantities $I,M,N$, it has been shown that they are related by the strict equality~\cite{extractingdynamicalquantumresources}
\begin{equation}
    I=M+N,
\end{equation}
which implies that Markovian and non-Markovian contributions to the total temporal correlations in $\mathbf{T}_{\hat{n}}$ necessarily come at the cost of each other. The situation for $\overline{I}_{\hat{m}},\overline{M}_{\hat{m}},\overline{N}_{\hat{m}}$ is slightly more nuanced: the superprocess $\mathbf{Z}^{M}_{\hat{n}\hat{n}}$ optimising $\overline{M}_{\hat{m}}$ is in general not the same as the superprocess $\mathbf{Z}^{N}_{\hat{n}\hat{n}}$ optimising $\overline{N}_{\hat{m}}$. While this implies that the strict equality above fails to hold for these monotones, they still satisfy subadditivity:
\begin{theorem}[Subadditivity of $\overline{M}$ and $\overline{N}$] \label{thm:subadditivity}
For any process $\mathbf{T}_{\hat{n}}$, the quantifiers $\overline{I}_{\hat{m}},\overline{M}_{\hat{m}},\overline{N}_{\hat{m}}$ satisfy
    \begin{gather}
\overline{I}_{\hat{m}}(\mathbf{T}_{\hat{n}}) \leq    \overline{M}_{\hat{m}}(\mathbf{T}_{\hat{n}}) + \overline{N}_{\hat{m}}(\mathbf{T}_{\hat{n}}).     
\end{gather}
\end{theorem}
\begin{proof}
Consider the difference
 \begin{align}
     &\overline{I}_{\hat{m}}(\mathbf{T}_{\hat{n}}) - \overline{M}_{\hat{m}}(\mathbf{T}_{\hat{n}}) \\
     =& \sup_{\mathbf{Z}_{\hat{n}\hat{n}} \in \mathsf{Z}_{\hat{n}\hat{n}}} I\Big( \llbracket \mathbf{T}_{\hat{n}} | \mathbf{Z}_{\hat{n}\hat{n}} | \mathbf{I}_{{\hat{n} \setminus \hat{m}}} \rrbracket \Big) \\ 
     -&  \sup_{\mathbf{Z}'_{\hat{n}\hat{n}} \in \mathsf{Z}_{\hat{n}\hat{n}}} M\Big( \llbracket \mathbf{T}_{\hat{n}} | \mathbf{Z}'_{\hat{n}\hat{n}}  | \mathbf{I}_{{\hat{n} \setminus \hat{m}}} \rrbracket\Big) \\
     \label{eqn::subadditivity}
     =& \sup_{\mathbf{Z}_{\hat{n}\hat{n}} \in \mathsf{Z}_{\hat{n}\hat{n}}} S\left(\llbracket \mathbf{T}_{\hat{n}} | \mathbf{Z}_{\hat{n}\hat{n}} | \mathbf{I}_{{\hat{n} \setminus \hat{m}}} \rrbracket \| \llbracket \mathbf{T}_{\hat{n}} | \mathbf{Z}_{\hat{n}\hat{n}} | \mathbf{I}_{{\hat{n} \setminus \hat{m}}} \rrbracket^{\text{marg}}\right) \\
     -&
     \sup_{\mathbf{Z}'_{\hat{n}\hat{n}} \in \mathsf{Z}_{\hat{n}\hat{n}}} S\left(\llbracket \mathbf{T}_{\hat{n}}| \mathbf{Z}'_{\hat{n}\hat{n}} | \mathbf{I}_{{\hat{n} \setminus \hat{m}}} \rrbracket^{\text{Mkv}} \|  \llbracket \mathbf{T}_{\hat{n}} | \mathbf{Z}'_{\hat{n}\hat{n}} | \mathbf{I}_{{\hat{n} \setminus \hat{m}}} \rrbracket^{\text{marg}} \right),
 \end{align}
 where $\cdot^{{\text{Mkv}}}$ and $\cdot^{{\text{marg}}}$ yield the tensor product of marginals according to Eq.~\eqref{eq:marginals}, respectively. Assuming that $\mathsf{Z}_{\hat{n}\hat{n}}$ is compact (which holds for all sets $\mathsf{Z}_{\hat{n}\hat{n}}$ we consider), we denote by $\mathbf{Z}_{\hat{n}\hat{n}}^*$ the superprocess that maximises the first term in Eq.~\eqref{eqn::subadditivity}. Then we obtain
 \begin{align}
     &\overline{I}_{\hat{m}}(\mathbf{T}_{\hat{n}}) - \overline{M}_{\hat{m}}(\mathbf{T}_{\hat{n}}) \\
     \notag
     \leq & S\left(\llbracket \mathbf{T}_{\hat{n}} | \mathbf{Z}^*_{\hat{n}\hat{n}} | \mathbf{I}_{{\hat{n} \setminus \hat{m}}} \rrbracket \| \llbracket \mathbf{T}_{\hat{n}} | \mathbf{Z}^*_{\hat{n}\hat{n}} | \mathbf{I}_{{\hat{n} \setminus \hat{m}}} \rrbracket^{\text{marg}}\right)  \\ 
     -& 
     S\left(\llbracket \mathbf{T}_{\hat{n}}| \mathbf{Z}^*_{\hat{n}\hat{n}} | \mathbf{I}_{{\hat{n} \setminus \hat{m}}} \rrbracket^{\text{Mkv}} \|  \llbracket \mathbf{T}_{\hat{n}} | \mathbf{Z}^*_{\hat{n}\hat{n}} | \mathbf{I}_{{\hat{n} \setminus \hat{m}}} \rrbracket^{\text{marg}} \right) \\
    \notag
     =& -S\left(\llbracket \mathbf{T}_{\hat{n}} | \mathbf{Z}^*_{\hat{n}\hat{n}} | \mathbf{I}_{{\hat{n} \setminus \hat{m}}} \rrbracket\right) + S\left(\llbracket \mathbf{T}_{\hat{n}} | \mathbf{Z}^*_{\hat{n}\hat{n}} | \mathbf{I}_{{\hat{n} \setminus \hat{m}}} \rrbracket^{\text{marg}} \right) \\
     +& S\left(\llbracket \mathbf{T}_{\hat{n}} | \mathbf{Z}^*_{\hat{n}\hat{n}} | \mathbf{I}_{{\hat{n} \setminus \hat{m}}} \rrbracket^{\text{Mkv}} \right) - S\left(\llbracket \mathbf{T}_{\hat{n}} | \mathbf{Z}^*_{\hat{n}\hat{n}} | \mathbf{I}_{{\hat{n} \setminus \hat{m}}} \rrbracket^{\text{marg}} \right) \\
     =& S\left(\llbracket \mathbf{T}_{\hat{n}} | \mathbf{Z}^*_{\hat{n}\hat{n}} | \mathbf{I}_{{\hat{n} \setminus \hat{m}}} \rrbracket^{\text{Mkv}} \right) -S\left(\llbracket \mathbf{T}_{\hat{n}} | \mathbf{Z}^*_{\hat{n}\hat{n}} | \mathbf{I}_{{\hat{n} \setminus \hat{m}}} \rrbracket \right)\\
     \leq & \sup_{\mathbf{Z}_{\hat{n}\hat{n}} \in \mathsf{Z}_{\hat{n}\hat{n}}} \left[S\left(\llbracket \mathbf{T}_{\hat{n}} | \mathbf{Z}_{\hat{n}\hat{n}} | \mathbf{I}_{{\hat{n} \setminus \hat{m}}} \rrbracket^{\text{Mkv}} \right) -S\left(\llbracket \mathbf{T}_{\hat{n}} | \mathbf{Z}_{\hat{n}\hat{n}} | \mathbf{I}_{{\hat{n} \setminus \hat{m}}} \rrbracket \right) \right] \\
     =& \sup_{\mathbf{Z}_{\hat{n}\hat{n}} \in \mathsf{Z}_{\hat{n}\hat{n}}} \left[S\left(\ \llbracket \mathbf{T}_{\hat{n}} | \mathbf{Z}_{\hat{n}\hat{n}} | \mathbf{I}_{{\hat{n} \setminus \hat{m}}} \rrbracket  \| \llbracket \mathbf{T}_{\hat{n}} | \mathbf{Z}_{\hat{n}\hat{n}} | \mathbf{I}_{{\hat{n} \setminus \hat{m}}} \rrbracket^{\text{Mkv}} \right)\right]\\
     =& \overline{N}_{\hat{m}}(\mathbf{T}_{\hat{n}}).
 \end{align}
The first equality comes from explicitly writing out the relative entropy between a quantum state and its own marginals. The second inequality holds because the value for a particular $\mathbf{Z}_{\hat{n}\hat{n}} \in \mathsf{Z}_{\hat{n}\hat{n}}$ cannot be greater than the supremum over $ \mathsf{Z}_{\hat{n}\hat{n}}$.
\end{proof}

Thm.~\ref{thm:subadditivity} establishes a trade-off between different resources; optimisation of the total information $\overline{I}_{\hat{m}}$ will generally not also optimise the non-Markovianity $\overline{N}_{\hat{m}}$ and the Markovianity $\overline{M}_{\hat{m}}$. Even if all the applied control is unitary (as in DD), temporal coarse-graining forces resources of the type not optimised for to be sacrificed irreversibly. The bound is tight when no allowed superprocess can increase $I,M,N$ at temporal resolution ${\hat{n}}$ for a given process $\mathbf{T}_{\hat{n}}$. However, tightness in more general circumstances is an open question.

\section{Composition Structure}
\begin{figure}
  \centering
  \subfloat[Sequential Composition.]{
    \includegraphics[width=0.95\linewidth]{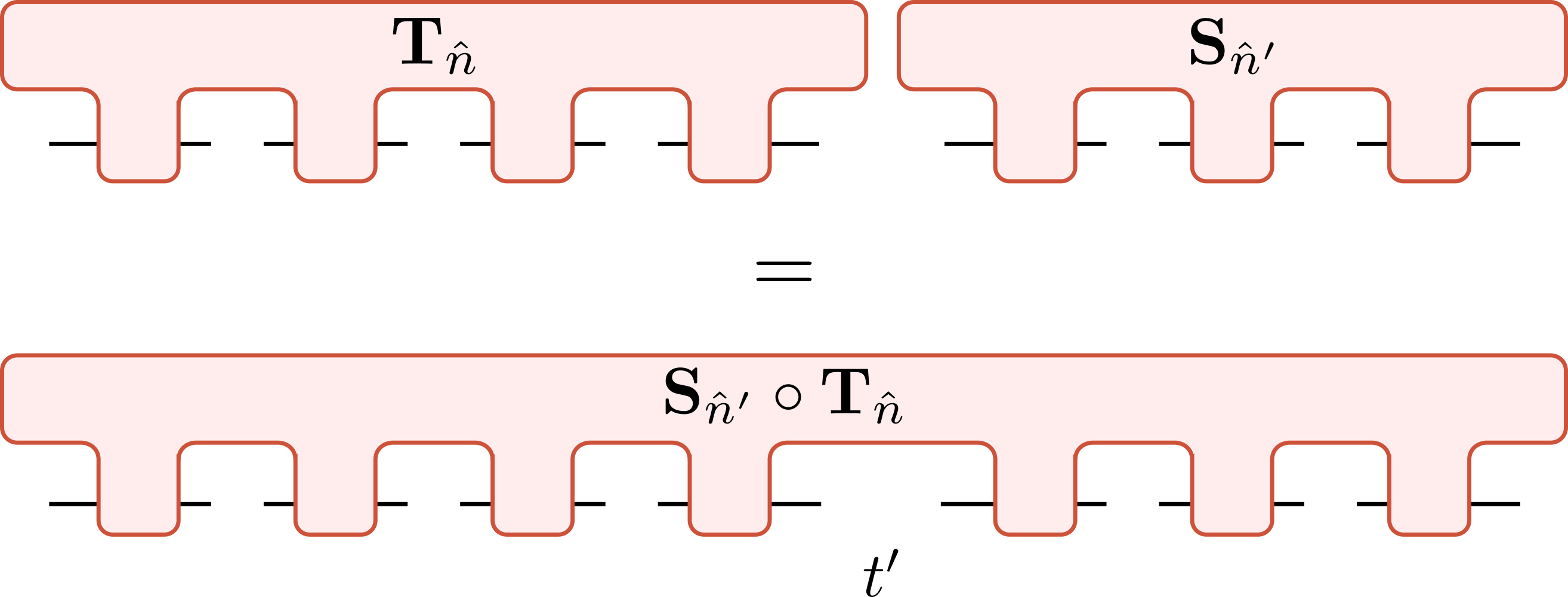}
    \label{fig::seq_comp}
  }\hfill
  \subfloat[Parallel Composition.]{
    \includegraphics[width=0.95\linewidth]{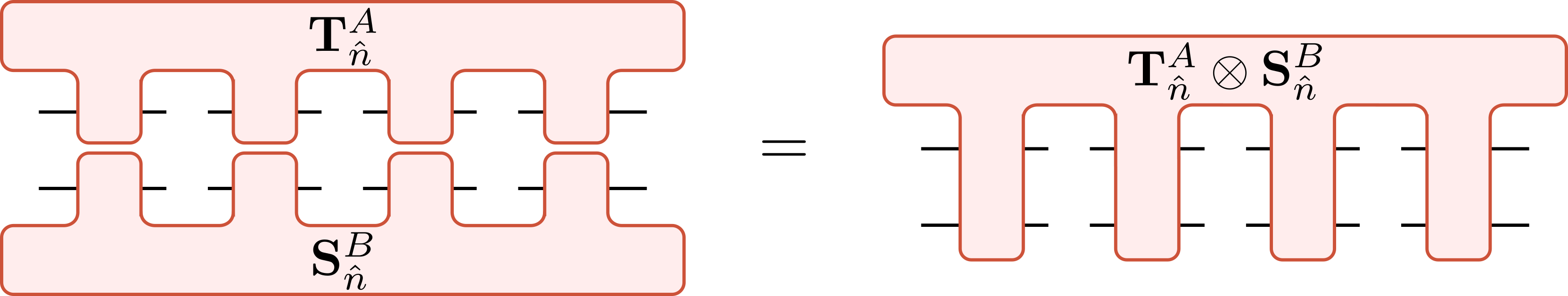}
    \label{fig::par_comp}
  }
  \caption{\textbf{Composition of quantum processes. (a)} Composing two processes $\mathbf{T}_{\hat{n}}$ and $\mathbf{S}_{\hat{n}'}$ in sequence yields a new time slot (denoted by $t'$) such that $\mathbf{S}_{\hat{n}'} \circ \mathbf{T}_{\hat{n}}$ is defined on $\hat{n} \cup \{t'\} \cup \hat{n}'$. \textbf{(b)} Composing two processes $\mathbf{T}_{\hat{n}}^A$ and $\mathbf{S}_{\hat{n}}^B$ in parallel yields a new process $\mathbf{T}_{\hat{n}}^A \otimes \mathbf{S}_{\hat{n}}^B$ on the same set of times $\hat{n}$. }
  \label{fig::composition}
\end{figure}
\label{sec:compositionstructure}
Thus far, we have predominantly considered the question of noise reduction when a single process tensor $\mathbf{T}_{\hat{n}}$ is available, and we have seen that $\overline{I}_\emptyset$ can be used as a quantifier of noise reduction. In this case, noise reduction is possible because the initial object is a multi-slot comb on which control operations can be performed, while the desired final object is a channel that is less noisy than if \textit{no} control operations had been performed. As such, denoising can be understood as a temporal \textit{resource distillation}~\cite{extractingdynamicalquantumresources}, where the resource of fine-grained temporal access to the process is used to distil all correlations present between different times into the final channel, which has zero slots. 

However, processes a priori not only allow temporal resource distillation, but one could -- similar to the case of quantum states -- be given multiple copies of the same process, or of different processes (see Fig.~\ref{fig::composition}). In this case, for the problem of noise reduction, the question arises how the resource content -- measured by $\overline{I}_{\hat{n}}$ -- changes as multitime process are combined? Unlike quantum states, multitime processes can be combined in many non-equivalent ways~\cite{theoreticalframeworkforquantumnetworks, Hirche2023}. Here, we study serial and parallel composition and focus on the investigation of $\overline{I}$; similar arguments can be applied to $\overline{M}$ and $\overline{N}$.

\subsection{Sequential Composition}
First, we consider sequential composition of a process $\mathbf{T}_{\hat{n}}$ on $n$ intermediate times with another process $\mathbf{S}_{\hat{n}'}$ on $n'$ intermediate times. In slight abuse of notation, we denote the resulting process by $\mathbf{S}_{\hat{n}'} \circ \mathbf{T}_{\hat{n}}$. We emphasise that we do not 'connect' the final output line of $\mathbf{T}_{\hat{n}}$ with the initial input line of $\mathbf{S}_{\hat{n}'}$. As a consequence, this composition induces a `new' time slot in between the two processes, which we denote by $t'$ (see Fig.~\ref{fig::seq_comp}). Consequently, the resulting quantifier is of the form $\overline{I}_{(\hat{m}, t' , \hat{m}')}$. We denote the corresponding IQI superprocesses on the resulting process tensor by $\mathbf{Z}_{(\hat{n}, \hat{n}')(\hat{n} , \hat{n}')}$, while the IQI superprocesses on the individual process tensors are denoted by $\mathbf{Z}_{\hat{n}\hat{n}}$ and $\mathbf{Z}_{\hat{n}'\hat{n}'}$, respectively. With this, we have the following Proposition:
\begin{proposition}[Sequential Composition] \label{prop:sequentialcomp}
    $\overline{I}_{\hat{m}}$ is additive in IQI under sequential composition.
\end{proposition}
\begin{proof}
Let $\mathbf{T}_{\hat{n}}$ and $\mathbf{S}_{\hat{n}'}$ be arbitrary processes. Then,\begin{align}
    &\overline{I}_{(\hat{m}, t' , \hat{m}')}( \mathbf{S}_{\hat{n}'} \circ \mathbf{T}_{\hat{n}} ) \\ 
    =&  \sup_{\mathbf{Z}_{(\hat{n} , \hat{n}')(\hat{n} , \hat{n}')} } I\Big( \llbracket \mathbf{S}_{\hat{n}'} \circ \mathbf{T}_{\hat{n}} | \mathbf{Z}_{(\hat{n}', \hat{n}) (\hat{n} , \hat{n}')} | \mathbf{I}_{{(\hat{n} , \hat{n}')}\setminus{(\hat{m}, t' , \hat{m}')}} \rrbracket  \Big) \\
    =& \sup_{\mathbf{Z}_{\hat{n}'\hat{n}'} \circ \mathbf{Z}_{\hat{n}\hat{n}}  } I\Big( \llbracket \mathbf{S}_{\hat{n}'} \circ \mathbf{T}_{\hat{n}} | \mathbf{Z}_{\hat{n}'\hat{n}'} \circ \mathbf{Z}_{\hat{n}\hat{n}}  | \mathbf{I}_{{(\hat{n} , \hat{n}')}\setminus{(\hat{m}, t' , \hat{m}')}} \rrbracket  \Big) \\
    =& \sup_{\mathbf{Z}_{\hat{n}'\hat{n}'} , \mathbf{Z}_{\hat{n}\hat{n}}  } I\Big( \llbracket \mathbf{S}_{\hat{n}'}  | \mathbf{Z}_{\hat{n}'\hat{n}'}  | \mathbf{I}_{{ \hat{n}'}\setminus{ \hat{m}'}} \rrbracket \otimes \llbracket \mathbf{T}_{\hat{n}} |  \mathbf{Z}_{\hat{n}\hat{n}}  | \mathbf{I}_{{\hat{n}}\setminus{\hat{m}}} \rrbracket  \Big) \\
    =& \sup_{\mathbf{Z}_{\hat{n}'\hat{n}'} } I\Big( \llbracket \mathbf{S}_{\hat{n}\prime} | \mathbf{Z}_{\hat{n}'\hat{n}'} | \mathbf{I}_{ \hat{n}' \setminus \hat{m}'} \rrbracket  \Big) + \sup_{\mathbf{Z}_{\hat{n}\hat{n}} } I\Big( \llbracket \mathbf{T}_{\hat{n}} | \mathbf{Z}_{\hat{n}\hat{n}} | \mathbf{I}_{{\hat{n}}\setminus\hat{m}} \rrbracket  \Big) \\
    =&  \overline{I}_{\hat{m}'}( \mathbf{S}_{\hat{n}'} ) + \overline{I}_{\hat{m}}( \mathbf{T}_{\hat{n}} ) .
\end{align}
The second equality is due to the fact that superprocesses in IQI are entirely uncorrelated between different times. The fourth equality holds due to the additivity of entropy for uncorrelated subsystems and because sequential composition does not impart correlations. We emphasise that in the above derivation, we have assumed that no coarse-graining happens at the additional time slot at $t'$, such that the corresponding $\overline{I}_{(\hat{m}, t', \hat{m}')}$ quantifies the correlation content for processes where the slot at $t'$ is still `open', i.e., it has not been coarse-grained over.
\end{proof}
As a consequence of this Proposition, if directly follows that if $\mathbf{S}_{\hat{n}'}$ has \textit{no} temporal correlations, i.e., $\overline{I}_{\hat{n}'} =0$, then composition with it cannot add any resource value. Additionally, since the memory subsystem of a quantum process is inaccessible from its process tensor, sequential composition cannot induce non-Markovianity across the bipartition between times $\hat{n}$ and $\hat{n}'$ (see Fig.~\ref{fig::seq_comp}).

\subsection{Parallel Composition}
We move on to parallel composition. Here, we only consider parallel composition of processes that have the \textit{same} number of slots, and the slots are located at the same times. Combining them in parallel does thus not change the number of slots -- unlike in the case of sequential composition. However, it increases the system size in the individual slots since more subsystems are accessed at any given time (see Fig.~\ref{fig::par_comp}). For this scenario, the situation is less straightforward than for sequential composition: in parallel composition, the `whole can be greater than the sum of its parts', i.e., resources are superadditive:
\begin{proposition}[Superadditivity Under Parallel Composition] \label{prop:superadditivity}
    $\overline{I}_{\hat{m}}$ is superadditive in IQI under parallel composition.
\end{proposition}
\begin{proof}
We label the subsystems of processes $\mathbf{T}_{\hat{n}}$ and $\mathbf{S}_{\hat{n}}$ as $A$ and $B$ respectively. With this, we have
\begin{align}
    &\overline{I}_{\hat{m}}( \mathbf{T}_{\hat{n}} \otimes \mathbf{S}_{\hat{n}} ) \\
    =& \sup_{\mathbf{Z}^{AB}_{\hat{n}\hat{n}} } I\Big( \llbracket \mathbf{T}^{A}_{\hat{n}} \otimes \mathbf{S}^{B}_{\hat{n}} | \mathbf{Z}^{AB}_{\hat{n}\hat{n}}| \mathbf{I}^{AB}_{{\hat{n}}\setminus{\hat{m}}} \rrbracket  \Big) \\
    \geq & \sup_{\mathbf{Z}^{A}_{\hat{n}\hat{n}},\mathbf{Z}'^{B}_{\hat{n}\hat{n}} } I\Big( \llbracket \mathbf{T}^{A}_{\hat{n}} \otimes \mathbf{S}^{B}_{\hat{n}} | \mathbf{Z}^{A}_{\hat{n}\hat{n}} \otimes \mathbf{Z}'^{B}_{\hat{n}\hat{n}}| \mathbf{I}^{AB}_{{\hat{n}}\setminus{\hat{m}}} \rrbracket  \Big) \\
    =& \sup_{\mathbf{Z}^{A}_{\hat{n}\hat{n}} } I\Big( \llbracket \mathbf{T}^{A}_{\hat{n}}  | \mathbf{Z}^{A}_{\hat{n}\hat{n}}| \mathbf{I}^{A}_{{\hat{n}}\setminus{\hat{m}}} \rrbracket  \Big) + \sup_{\mathbf{Z}'^{B}_{\hat{n}\hat{n}} } I\Big( \llbracket \mathbf{S}^{B}_{\hat{n}}  | \mathbf{Z}'^{B}_{\hat{n}\hat{n}}| \mathbf{I}^{B}_{{\hat{n}}\setminus{\hat{m}}} \rrbracket  \Big) \\
    =&\overline{I}_{\hat{m}}( \mathbf{T}_{\hat{n}}) + \overline{I}_{\hat{m}}( \mathbf{S}_{\hat{n}}).
\end{align}
The inequality holds because the supremum over product superprocesses $\mathbf{Z}^{A}_{\hat{n}\hat{n}} \otimes \mathbf{Z}'^{B}_{\hat{n}\hat{n}}$ (instead of over \textit{all} superprocesses $\mathbf{Z}^{AB}_{\hat{n}\hat{n}}$ that IQI permits) is not necessarily optimal. The second equality holds due to the additivity of entropy for uncorrelated subsystems. 
\end{proof}
Fundamentally, the superadditivity property of Prop.~\ref{prop:superadditivity} holds since having multiple simultaneous processes provides more options for how to manipulate them, given that IQI does restrict manipulations temporally, but \textit{not} spatially. That is,  IQI superprocesses are local in time but allow for the concurrent manipulations of distinct subsystems in each slot. 

Due to superadditivity, the maximum values of $\overline{I}$, $\overline{M}$, and $\overline{N}$ that can be practically extracted from a multitime process in IQI is defined by the asymptotic limit of many copies of that process, not the single-copy value. This is reminiscent of the situation of noisy quantum channel coding, where quantum capacity is defined in the asymptotic limit due to the superadditivity of coherent information. As such, we expect that regularisations of $\overline{I}_{\hat{m}},\overline{M}_{\hat{m}},\overline{N}_{\hat{m}}$ possess useful operational interpretations for for similar asymptotic tasks.

Naturally, superadditivity should not allow one to generate resources for free, i.e., create correlations from processes that do not contain any correlations to start with [see Eq.~\eqref{eq:fullyuncorrelated}]. This fact is expressed in the following Proposition:
\begin{proposition}[Invariance for fully uncorrelated Processes]
    $\overline{I}_{\hat{m}}$ is invariant in IQI under parallel composition with processes $\mathbf{S}_{\hat{n}}$ satisfying Eq.~\eqref{eq:fullyuncorrelated}.
\end{proposition}
\begin{proof}
From Prop.~\ref{prop:superadditivity} we know that $ \overline{I}_{\hat{m}}( \mathbf{T}_{\hat{n}} \otimes \mathbf{S}_{\hat{n}} ) \geq \overline{I}_{\hat{m}}( \mathbf{T}_{\hat{n}}) + \overline{I}_{\hat{m}}( \mathbf{S}_{\hat{n}})$ for all processes $\mathbf{S}_{\hat{n}}$. We thus only need to show that $ \overline{I}_{\hat{m}}( \mathbf{T}_{\hat{n}} \otimes \mathbf{S}_{\hat{n}} ) \leq \overline{I}_{\hat{m}}( \mathbf{T}_{\hat{n}}) + \overline{I}_{\hat{m}}( \mathbf{S}_{\hat{n}})$ whenever $\mathbf{S}_{\hat{n}}$ is fully uncorrelated. Processes without temporal correlations are of the form $(\mathbbm{1}/d_s \otimes \rho_1)\otimes \dots \otimes (\mathbbm{1}/d_s \otimes \rho_{n+1})$ (see Eq.~\eqref{eq:fullyuncorrelated}). As a consequence, $ \llbracket \mathbf{S}^{B}_{\hat{n}} | \mathbf{Z}^{AB}_{\hat{n}\hat{n}}  \rrbracket = \mathbf{Z}^{\prime A}_{\hat{n}\hat{n}}$ is a superprocess on $A$ that lies in in IQI, and we have $ \llbracket \mathbf{S}^{B}_{\hat{n}} | \mathsf{Z}^{AB}_{\hat{n}\hat{n}} \rrbracket \subseteq \mathsf{Z}^{A}_{\hat{n}\hat{n}}$. Here, $B$ is subsumed into a larger effective $A$ subsystem, and $ \mathbf{Z}^{\prime A}_{\hat{n}\hat{n}}$ is still unable to transmit correlations on this effective $A$ subsystem through time. With this, we can show:
\begin{align}
    \overline{I}_{\hat{m}}( \mathbf{T}_{\hat{n}} \otimes \mathbf{S}_{\hat{n}} ) =& \sup_{\mathbf{Z}^{AB}_{\hat{n}\hat{n}} } I\Big( \llbracket \mathbf{T}^{A}_{\hat{n}} \otimes \mathbf{S}^{B}_{\hat{n}} | \mathbf{Z}^{AB}_{\hat{n}\hat{n}}| \mathbf{I}^{AB}_{{\hat{n}}\setminus{\hat{m}}} \rrbracket  \Big) \\
    =&  \sup_{\mathbf{Z}^{\prime A}_{\hat{n}\hat{n}} } I\Big( \llbracket \mathbf{T}^{A}_{\hat{n}} | \mathbf{Z}^{\prime A}_{\hat{n}\hat{n}}| \mathbf{I}^{A}_{{\hat{n}}\setminus{\hat{m}}} \rrbracket  \Big)  \\
    \leq & \sup_{\mathbf{Z}^{A}_{\hat{n}\hat{n}} \in \mathsf{Z}_{\hat{n}\hat{n}}^{A}} I\Big( \llbracket \mathbf{T}^{A}_{\hat{n}}  | \mathbf{Z}^{A}_{\hat{n}\hat{n}}| \mathbf{I}^{A}_{{\hat{n}}\setminus{\hat{m}}} \rrbracket  \Big)  \\
    =&\overline{I}_{\hat{m}}( \mathbf{T}_{\hat{n}}),
\end{align}
where we have used in the second equality that resource-free processes in IQI do not contain any correlations in time, and thus do not lead to superprocesses $\mathbf{Z}_{\hat{n}\hat{n}}^{\prime A}$ that lie outside of what is possible within IQI. 
\end{proof}

\section{Correspondence with Generalised Comb Divergences} \label{sec:correspondence}
The resource quantifiers $\overline{I}_{\hat{m}}(\mathbf{T}_{\hat{n}}),     \overline{N}_{\hat{m}}(\mathbf{T}_{\hat{n}})$ and $\overline{M}_{\hat{m}}(\mathbf{T}_{\hat{n}})$ introduced in Thm.~\ref{thm::irrev} decrease monotonically under the superprocesses of IQI (or, more generally, the superprocesses within the respective set $\mathsf{Z}_{\hat{n}{\hat{n}}}$) and temporal coarse-graining, thus satifsying the basic desideratum of a resource quantifier with respect to the superprocesses in $\mathsf{Z}_{\hat{n}\hat{n}}$. When $\mathsf{Z}_{\hat{n}\hat{n}}$ corresponds to superprocesses in IQI, $\overline{I}_{\hat{m}}$ ($\overline{N}_{\hat{m}}$) quantifies the maximal possible distance (as measured by the relative entropy $S$) of a process $\mathbf{T}_{\hat{n}}$ to its fully uncorrelated (Markovian) counterpart under superprocesses in IQI, while $\overline{M}_{\hat{m}}$ denotes the corresponding maximal distance between the Markovian and fully uncorrelated versions of $\mathbf{T}_{\hat{n}}$ [see Eq.~\eqref{eq:Imonotoneresolution}]. 

In a similar vein, the \textit{generalised comb divergence} $\widetilde{D}$ has been introduced in Ref.~\cite{zambon2024processtensordistinguishabilitymeasures} as a mathematically well-behaved (pseudo-) distance measure between quantum combs that decreases monotonically under \textit{all} superprocesses $\mathbf{Z}_{\hat{n}\hat{m}}$; here, we relate these two concepts and show that the quantifiers $\overline{I}_{\hat{m}}, \overline{N}_{\hat{m}}$ and $\overline{M}_{\hat{m}}$ naturally induce a class of `reachable' comb divergences by restricting the maximisation appearing in the computation of $\widetilde{D}$ to limited sets of superprocesses. In particular, we show that $\overline{I}_{\emptyset}(\mathbf{T}_{\hat{n}}) = \overline{M}_{\emptyset}(\mathbf{T}_{\hat{n}})$ can be understood as a lower bound on the comb divergences of Ref.~\cite{zambon2024processtensordistinguishabilitymeasures} as well as the corresponding reachable comb divergences we introduce here. 

Intuitively, the generalised comb divergence of two combs $\mathbf{T}_{\hat{n}}$ and $\mathbf{R}_{\hat{n}}$ quantifies the maximum relative entropy between quantum states that can be obtained from them via a superprocess. In detail, given two combs $\mathbf{T}_{\hat{n}}$ and $\mathbf{R}_{\hat{n}}$ defined on the same set of times, let $\widetilde{\mathsf{Y}}_{\hat{n}\emptyset}$ be the set of all superprocesses that map both of these combs to quantum states, i.e., $\llbracket \mathbf{T}_{\hat{n}}| \widetilde{\mathbf{Y}}_{\hat{n}\emptyset}\rrbracket, \llbracket \mathbf{R}_{\hat{n}}| \widetilde{\mathbf{Y}}_{\hat{n}\emptyset}\rrbracket \in \mathcal{L}(\mathcal{H}_{\text{in}_\text{f}} \otimes \mathcal{H}_{\text{a}_\text{f}})$ for all $\widetilde{\mathbf{Y}}_{\hat{n}\emptyset} \in \widetilde{\mathsf{Y}}_{\hat{n}\emptyset}$, where $\mathcal{H}_{\text{a}_\text{f}}$ is some arbitrary additional `reference' space that  $\mathbf{T}_{\hat{n}}$ and $\mathbf{R}_{\hat{n}}$ do not act on. Then the generalised comb divergence between $\mathbf{T}_{\hat{n}}$ and $\mathbf{R}_{\hat{n}}$ is given by 
\begin{gather}
\label{eqn::gen_div_og}
    \widetilde{D}\big(  \mathbf{T}_{\hat{n}} \big\Vert \mathbf{R}_{\hat{n}} \big) := \sup_{\widetilde{\mathbf{Y}}_{\hat{n}\emptyset} \in \widetilde{\mathsf{Y}}_{\hat{n}\emptyset}} S\Big( \llbracket \mathbf{T}_{\hat{n}}| \widetilde{\mathbf{Y}}_{\hat{n}\emptyset}\rrbracket \Big\Vert \llbracket \mathbf{R}_{\hat{n}}|\widetilde{\mathbf{Y}}_{\hat{n}\emptyset}\rrbracket \Big).
\end{gather}
As shown in Ref.~\cite{zambon2024processtensordistinguishabilitymeasures} $\widetilde{D}$ behaves monotonically under \textit{all} superprocesses $\mathbf{Z}_{\hat{n}\hat{m}}$, i.e., $\widetilde{D}\big(\mathbf{T}_{\hat{n}} \big\Vert \mathbf{R}_{\hat{n}} \big) \geq \widetilde{D}\big( \llbracket \mathbf{R}_{\hat{n}}| \mathbf{Z}_{\hat{n}\hat{m}}\rrbracket \big\Vert  \llbracket \mathbf{T}_{\hat{n}}| \mathbf{Z}_{\hat{n}\hat{m}}\rrbracket\big)$, for all superprocesses $\mathbf{Z}_{\hat{n}\hat{m}}$, thus satisfying a generalised data processing inequality . In contrast, the monotones introduced in this paper are only guaranteed to behave monotonically under a \textit{restricted} set of superprocesses, for example, those obtained from IQI and coarse-graining. 
To connect these monotones to the results of Ref.~\cite{zambon2024processtensordistinguishabilitymeasures}, we first adapt the generalised comb divergences to our setting, defining
\begin{equation}
\label{eqn::gen_div}
    D\big(  \mathbf{T}_{\hat{n}} \big\Vert \mathbf{R}_{\hat{n}} \big) := \sup_{\mathbf{Y}_{\hat{n}\emptyset} \in \mathsf{Y}_{\hat{n}\emptyset}} S\Big( \llbracket \mathbf{T}_{\hat{n}\emptyset}| \mathbf{Y}_{\hat{n}\emptyset}\rrbracket \Big\Vert \llbracket \mathbf{R}_{\hat{n}}|\mathbf{Y}_{\hat{n}\emptyset}\rrbracket \Big),
\end{equation}
where $\mathbf{Y}_{\hat{n}\emptyset}$ is a control comb that is such that $\llbracket \mathbf{T}_{\hat{n}}| \mathbf{Y}_{\hat{n}\emptyset}\rrbracket, \llbracket \mathbf{R}_{\hat{n}}|\mathbf{Y}_{\hat{n}\emptyset}\rrbracket \in \mathcal{L}(\mathcal{H}_{\text{in}_\text{f}} \otimes \mathcal{H}_{\text{a}_\text{f}} \otimes \mathcal{H}_{\text{out}_\text{i}} \otimes \mathcal{H}_{\text{a}_\text{i}})$ correspond to quantum \textit{channels}, respectively, in contrast to Eq.~\eqref{eqn::gen_div_og}, where the resulting objects corresponded to quantum \textit{states}. This difference in the definition of comb divergences is due to the fact that our primary concern here is the noise properties of the resultant channel. Consequently, the superprocess $\mathbf{Y}_{\hat{n}\emptyset}$ used in Eq.~\eqref{eqn::gen_div} may contain a \textit{pair} of auxiliary subsystems $\{\text{a}_\text{i}, \text{a}_\text{f}\}$ that $\mathbf{T}_{\hat{n}}$ and $\mathbf{R}_{\hat{n}}$ do not act on. 

Here, we show that for the case where no open slots remain the resource quantifiers $\overline{I}_{\emptyset}$ and $\overline{M}_{\emptyset}$ introduced in Thm.~\ref{thm::irrev} can be directly related to Eq.~\eqref{eqn::gen_div}, with the main difference that the supremum is taken over a limited set of control combs $\mathsf{Y}^\textup{reach}_{\hat{n}\emptyset}$ instead of arbitrary control combs (note that $\overline{N}_{\emptyset}$ vanishes when no open time slots remain). To this end, we first note that, given sets of superprocesses $\mathsf{Z}_{\hat{n}\hat{m}}$, possibly changing the temporal resolution, we can obtain the corresponding sets $\mathsf{Y}^\text{reach}_{\hat{n}\emptyset}$ of \textit{reachable} control combs as those that can be obtained from $\mathsf{Z}_{\hat{n}\hat{m}}$ via coarse-graining, i.e., 
\begin{gather}
\label{eqn::reachable}
    \mathsf{Y}^\text{reach}_{\hat{n}\emptyset} := \{ \mathbf{Y}_{\hat{n}\emptyset} = \llbracket \mathbf{Z}_{\hat{n}\hat{m}}| \mathbf{I}_{\hat{m}}\rrbracket \ \ \text{for some } \mathbf{Z}_{\hat{n}\hat{m}} \in \mathsf{Z}_{\hat{n}\hat{m}}\}.
\end{gather}
Throughout, we only consider \textit{compatible} sets $\mathsf{Z}_{\hat{n}\hat{m}}$ of superprocesses that are connected via temporal coarse-graining between different temporal resolutions, i.e., $\mathsf{Z}_{\hat{n}\hat{m}} = \llbracket \mathsf{Z}_{\hat{n}\hat{n}}|\mathbf{I}_{\hat{n}\setminus \hat{m}}\rrbracket$ for all $\hat{m}\subseteq \hat{n}$, and are closed under composition, i.e., $\mathsf{Z}_{\hat{n}\hat{m}} \mathsf{Z}_{\hat{m}\hat{\ell}} \subseteq \mathsf{Z}_{\hat{n}\hat{\ell}}$ for all $\hat{n}, \hat{m}$ and $\hat{\ell}$. With this, we have the following Lemma:

\begin{lemma}
Let $\mathsf{Z}_{\hat{n}\hat{m}}$ be compatible sets of superprocesses, and let $\mathsf{Y}^\textup{reach}_{\hat{n}\emptyset}$ be the corresponding sets of reachable control combs. Then the \textit{reachable} comb divergence 
\begin{gather}
    D^\textup{reach}\big( \mathbf{T}_{\hat{n}} \big\Vert \mathbf{R}_{\hat{n}} \big) := \sup_{\mathbf{Y}_{\hat{n}\emptyset} \in \mathsf{Y}^\textup{reach}_{\hat{n}\emptyset}} S\Big( \llbracket \mathbf{T}_{\hat{n}}| \mathbf{Y}_{\hat{n}\emptyset}\rrbracket \Big\Vert \llbracket \mathbf{R}_{\hat{n}}|\mathbf{Y}_{\hat{n}\emptyset}\rrbracket \Big),
\end{gather}
decreases monotonically under all superprocesses in $\mathsf{Z}_{\hat{n}\hat{m}}$.
\end{lemma}
\begin{proof}
    The proof follows a similar line of argument to that of Thm.~1 in Ref.~\cite{zambon2024processtensordistinguishabilitymeasures}. Let $\mathbf{Z}_{\hat{n}\hat{m}} \in \mathsf{Z}_{\hat{n}\hat{m}}$. We have 
    \begin{align}
    &D^\text{reach}\big( \llbracket \mathbf{T}_{\hat{n}}| \mathbf{Z}_{\hat{n}\hat{m}}\rrbracket \big\Vert \llbracket\mathbf{R}_{\hat{n}} |   \mathbf{Z}_{\hat{n}\hat{m}}\rrbracket\big) \\ 
    =& \sup_{\mathbf{Y}_{\hat{m}\emptyset} \in \mathsf{Y}^\text{reach}_{\hat{m}\emptyset}} S\Big( \llbracket \mathbf{T}_{\hat{n}}|\mathbf{Z}_{\hat{n}\hat{m}}| \mathbf{Y}_{\hat{m}\emptyset}\rrbracket \Big\Vert \llbracket \mathbf{R}_{\hat{n}}|\mathbf{Z}_{\hat{n}\hat{m}}|\mathbf{Y}_{\hat{m}\emptyset}\rrbracket \Big) \\
     =& \sup_{\mathbf{Z}'_{\hat{m}\hat{m}} \in \mathsf{Z}_{\hat{m}\hat{m}}} S\Big( \llbracket \mathbf{T}_{\hat{n}}|\mathbf{Z}_{\hat{n}\hat{m}}\mathbf{Z}'_{\hat{m}\hat{m}}|\mathbf{I}_{\hat{m}}\rrbracket \Big\Vert \llbracket \mathbf{R}_{\hat{m}}|\mathbf{Z}_{\hat{n}\hat{m}}\mathbf{Z}'_{\hat{m}\hat{m}}|\mathbf{I}_{\hat{m}}\rrbracket \Big) \\
    \label{eqn::reachDiv}
    \leq & \sup_{\mathbf{Z}'_{\hat{n}\hat{m}} \in \mathsf{Z}_{\hat{n}\hat{m}}} S\Big( \llbracket \mathbf{T}_{\hat{n}}|\mathbf{Z}'_{\hat{n}\hat{m}}|\mathbf{I}_{\hat{m}}\rrbracket \Big\Vert \llbracket \mathbf{R}_{\hat{m}}|\mathbf{Z}'_{\hat{n}\hat{m}}|\mathbf{I}_{\hat{m}}\rrbracket \Big) \\
    \label{eqn::reachDiv2}
    = & \sup_{\mathbf{Y}_{\hat{n}\emptyset} \in \mathsf{Y}^\text{reach}_{\hat{n}\emptyset}} S\Big( \llbracket \mathbf{T}_{\hat{n}}|\mathbf{Y}_{\hat{n}\emptyset}\rrbracket \Big\Vert \llbracket \mathbf{R}_{\hat{m}}|\mathbf{Y}_{\hat{n}\emptyset}\rrbracket \Big) \\
    =& D^\textup{reach}\big( \mathbf{T}_{\hat{n}} \big\Vert \mathbf{R}_{\hat{n}} \big),
    \end{align}
where, for the inequality and penultimate equality, we have used the fact that $\mathbf{Z}_{\hat{n}\hat{m}}\mathsf{Z}_{\hat{m}\hat{m}} \subseteq \mathsf{Z}_{\hat{n}\hat{m}}$, and that $\llbracket\mathsf{Z}_{\hat{n}\hat{m}}|\mathbf{I}_{\hat{m}}\rrbracket  = \mathsf{Y}_{\hat{n}\emptyset}^\text{reach}$, respectively, which both follow from the from the compatibility of the sets $\mathsf{Z}_{\hat{n}\hat{m}}$.
\end{proof}
As is evident from the proof of this Lemma [in particular in the step from Eq.~\eqref{eqn::reachDiv} to Eq.~\eqref{eqn::reachDiv2}] optimisation over sets of control combs, as required for the computation of reachable comb divergences, is equivalent to the optimisation over superprocesses $\mathsf{Z}_{\hat{n}\hat{m}}$ and subsequent temporal coarse-graining. Consequently, if $\mathsf{Z}_{\hat{n}\hat{m}}$ corresponds to superprocesses in IQI, then the quantifiers $\overline{I}_{\emptyset}(\mathbf{T}_{\hat{n}})$ and $ \overline{M}_{\emptyset}(\mathbf{T}_{\hat{n}})$ are upper bounded by the reachable comb divergence of $\mathbf{T}_{\hat{n}}$ with respect to its fully uncorrelated version $\mathbf{T}^{\text{marg}}_{\hat{n}}$. In particular, we have 
\begin{align}
&D^\text{reach}\big( \mathbf{T}_{\hat{n}} \big\Vert \mathbf{T}^{\text{marg}}_{\hat{n}}\big) \\
&= \sup_{\mathbf{Z}_{\hat{n}\hat{n}} \in \mathsf{Z}_{\hat{n}\hat{n}}} S\Big( \llbracket \mathbf{T}_{\hat{n}}| \mathbf{Z}_{\hat{n}\hat{n}}| \mathbf{I}_{{\hat{n}}}\rrbracket \Big\Vert \llbracket \mathbf{T}^{\text{marg}}_{\hat{n}}|\mathbf{Z}_{\hat{n}\hat{n}}|\mathbf{I}_{\hat{n}}\rrbracket \Big) \label{eq:noncommutativ}\\
&\geq \sup_{\mathbf{Z}_{\hat{n}\hat{n}} \in \mathsf{Z}_{\hat{n}\hat{n}}} S\Big( \llbracket \mathbf{T}_{\hat{n}}| \mathbf{Z}_{\hat{n}\hat{n}}| \mathbf{I}_{{\hat{n}}}\rrbracket \Big\Vert \llbracket \mathbf{T}_{\hat{n}}|\mathbf{Z}_{\hat{n}\hat{n}}|\mathbf{I}_{\hat{n}}\rrbracket^{\text{marg}} \Big), \label{eq:noncommutativityinequality}\\
&= \sup_{\mathbf{Z}_{\hat{n}\hat{n}} \in \mathsf{Z}_{\hat{n}\hat{n}}} I\Big( \llbracket \mathbf{T}_{\hat{n}}| \mathbf{Z}_{\hat{n}\hat{n}}| \mathbf{I}_{{\hat{n}}}\rrbracket \Big)\\
&= \overline{I}_{\emptyset}(\mathbf{T}_{\hat{n}}) = \overline{M}_{\emptyset}(\mathbf{T}_{\hat{n}}),
\end{align}
where we have used that the closest (in terms of the relative entropy) fully uncorrelated process to $\llbracket \mathbf{T}_{\hat{n}}| \mathbf{Z}_{\hat{n}\hat{n}}| \mathbf{I}_{{\hat{n}}}\rrbracket$ is the product of that process' own marginals, $\llbracket \mathbf{T}_{\hat{n}}| \mathbf{Z}_{\hat{n}\hat{n}}| \mathbf{I}_{{\hat{n}}}\rrbracket^{\text{marg}}$, yielding the inequality in the third line. Additionally, note that $\overline{I}$ and $\overline{M}$ coincide for quantum channels, i.e., when no open slots remain, leading to the final equality. Having established that $\overline{I}_{\emptyset}$ lower bounds $D^\text{reach}$, and given that generalised comb divergences upper bound reachable ones, we arrive at the hierarchy
\begin{gather}
    \overline{I}_\emptyset = \overline{M}_\emptyset \leq D^\text{reach}\big( \mathbf{T}_{\hat{n}} \big\Vert \mathbf{T}^{\text{marg}}_{\hat{n}}\big) \leq D\big( \mathbf{T}_{\hat{n}} \big\Vert \mathbf{T}^{\text{marg}}_{\hat{n}}\big).
\end{gather}
Reachable comb divergences can be understood as `pickier' resource quantifiers than generalised comb divergences, because they only consider transformations that are explicitly allowed within our specified setup.

The more general connection to reachable/generalised comb divergences for the case $\overline{I}_{\hat{m}}, \overline{M}_{\hat{m}}$ and $\overline{N}_{\hat{m}}$ for $\hat{m} \subset \hat{n}$ follows by introducing comb divergences in the same vein as in Eq.~\eqref{eqn::reachable}, but using $\mathsf{Y}^\text{reach}_{\hat{n}\hat{m}} = \llbracket \mathsf{Z}_{\hat{n}\hat{n}} | \mathbf{I}_{\hat{n}\setminus \hat{m}}\rrbracket$ and defining $D^{\text{reach}}_{\hat{m}}$ accordingly. With this, following the same line of reasoning as above, we obtain similar upper bounds on the introduced resource quantifiers:
\begin{gather}
\begin{split}
     &\overline{I}_{\hat{m}}(\mathbf{T}_{\hat{n}}) \leq D^{\text{reach}}_{\hat{m}}(\mathbf{T}_{\hat{n}}|| \mathbf{T}^{\text{marg}}_{\hat{n}}),\\ 
     &\overline{M}_{\hat{m}}(\mathbf{T}_{\hat{n}}) \leq D^{\text{reach}}_{\hat{m}}(\mathbf{T}^{\text{Mkv}}_{\hat{n}}|| \mathbf{T}^{\text{marg}}_{\hat{n}}),
     \\ \text{and} \ \ &\overline{N}_{\hat{m}}(\mathbf{T}_{\hat{n}}) \leq D^{\text{reach}}_{\hat{m}}(\mathbf{T}_{\hat{n}}|| \mathbf{T}^{\text{Mkv}}_{\hat{n}}).
\end{split}
\end{gather}
Whether there exist examples of $\mathbf{T}_{\hat{n}}$ where these quantities satisfy a strict inequality, or if they are always equal, is currently unknown.

\section{Dynamical Decoupling as Temporal Resource Distillation} \label{sec:numerics}
As mentioned throughout, the resource quantifiers we introduced are directly related to the amenability of a quantum process to dynamical decoupling, raising the question of how exactly these different resources are expended during the decoupling process. At face value, their physical interpretation is the maximum amount of relative entropy that a quantum process can exhibit -- with respect to its fully uncorrelated version ($\overline{I}_{\hat{m}}$), with respect to its Markovian version ($\overline{N}_{\hat{m}}$), or between its Markovian and fully uncorrelated version ($\overline{M}_{\hat{m}}$) -- at temporal resolution $\hat{m}$ under the available experimental control. Computation of these quantifiers is thus equivalent to the problem of finding an optimal quantum control policy over a non-Markovian quantum process. While the difficulty of deducing provably optimal solutions is not yet fully understood~\cite{Barry_2014,Ying2021}, there are numerous viable approximate approaches~\cite{werschnik2007quantumoptimalcontroltheory,Koch_2022,PhysRevA.84.022326,KHANEJA2005296}. 

Of particular relevance here is the multitimescale optimal dynamical decoupling~\cite{extractingdynamicalquantumresources} (MODD) scheme, originally introduced to find optimal sequences of operations that denoise a non-Markovian dynamics. Concretely, given $\mathbf{T}_{\hat{n}}$, the MODD protocol consists of applying a see-saw semidefinite program (SDP) at increasingly coarse timescales to search for the superprocess within IQI that yields the resultant channel $\mathbf{T}_\emptyset$ of highest mutual information between input and output. As a consequence, this numerical approach yields an approximation method for finding $\overline{I}_\emptyset$. It is not necessarily optimal, though, since i) the employed SDP maximises the largest eigenvalue of the Choi matrix of the resulting channel $\mathbf{T}_\emptyset$ as a proxy for the maximisation of its mutual information, not the mutual information itself and ii) the corresponding see-saw approach, iteratively optimising operations in each slot individually, is only guaranteed to converge to a local maximum~\cite{extractingdynamicalquantumresources}, not a global one. Nonetheless, as shown in Ref.~\cite{extractingdynamicalquantumresources}, this approximative computation of $\overline{I}_\emptyset$ yields a dramatic improvement over the $X,Z,X,Z$ pulse sequences -- leading to the superprocess $\mathbf{Z}^{\mathrm{DD}}_{\hat{n}\hat{n}}$ --  traditionally employed in dynamical decoupling for the task of denoising. This advantage suggests that non-Markovian processes tend to possess greater resource content than one would assume from the results of traditional DD alone.

Conducting the corresponding resource analysis, in Ref.~\cite{extractingdynamicalquantumresources}, it was claimed that since the MODD superprocess $\mathbf{Z}^{\text{MODD}}_{\hat{n}\hat{n}}$ and traditional DD superprocess $\mathbf{Z}^{DD}_{\hat{n}\hat{n}}$ satisfy $I\big( \llbracket \mathbf{T}_{\hat{n}} | \mathbf{Z}^{\text{MODD}}_{\hat{n}\hat{n}} | \mathbf{I}_{{\hat{n}}\setminus{\hat{m}}} \rrbracket \big) > I\big( \llbracket \mathbf{T}_{\hat{n}} | \mathbf{Z}^{\text{DD}}_{\hat{n}\hat{n}} | \mathbf{I}_{{\hat{n}}\setminus{\hat{m}}} \rrbracket \big)$, the traditional DD superprocess does not extract all resources present in $\mathbf{T}_{\hat{n}}$. However, since $I$ is not monotonic under the allowed control of IQI, i.e. it is not a meaningful resource quantifier, this intuitive conclusion is a priori not entirely justified, and does not follow directly from the premise. 

Instead, the correct statement would be that $\overline{I}_{\emptyset} ( \mathbf{T}_{\hat{n}} ) > I\big( \llbracket \mathbf{T}_{\hat{n}} | \mathbf{Z}^{\text{DD}}_{\hat{n}\hat{n}} | \mathbf{I}_{{\hat{n}}\setminus{\hat{m}}} \rrbracket \big)$ implies that not all resources have been utilised by traditional DD, as it signifies the existence of decoupling sequences that outperform $\mathbf{Z}_{\hat{n}\hat{n}}^\text{DD}$, and the corresponding possible advantage is quantified by $\overline{I}_{\emptyset} ( \mathbf{T}_{\hat{n}} ) - I\big( \llbracket \mathbf{T}_{\hat{n}} | \mathbf{Z}^{\text{DD}}_{\hat{n}\hat{n}} | \mathbf{I}_{{\hat{n}}\setminus{\hat{m}}} \rrbracket \big)$. Since, by construction,  $ \overline{I}_{\emptyset} (\mathbf{T}_{\hat{n}} ) \geq I\big( \llbracket \mathbf{T}_{\hat{n}} | \mathbf{Z}^{\text{MODD}}_{\hat{n}\hat{n}} | \mathbf{I}_{{\hat{n}}\setminus{\hat{m}}} \rrbracket \big)$, and Ref.~\cite{extractingdynamicalquantumresources} found that $I\big( \llbracket \mathbf{T}_{\hat{n}} | \mathbf{Z}^{\text{MODD}}_{\hat{n}\hat{n}} | \mathbf{I}_{{\hat{n}}\setminus{\hat{m}}} \rrbracket \big) > I\big( \llbracket \mathbf{T}_{\hat{n}} | \mathbf{Z}^{\text{DD}}_{\hat{n}\hat{n}} | \mathbf{I}_{{\hat{n}}\setminus{\hat{m}}} \rrbracket \big)$ for all processes $\mathbf{T}_{\hat{n}}$ that were analysed, the numerical results of Ref.~\cite{extractingdynamicalquantumresources} do indeed support the conclusion of resource wastage, but quantified by $\overline{I}_\emptyset$ rather than by $I$ directly. Performance using traditional DD underestimates the true potential for denoising, and the potential for improvement is lower bounded by $I\big( \llbracket \mathbf{T}_{\hat{n}} | \mathbf{Z}^{\text{MODD}}_{\hat{n}\hat{n}} | \mathbf{I}_{{\hat{n}}\setminus{\hat{m}}} \rrbracket \big) - I\big( \llbracket \mathbf{T}_{\hat{n}} | \mathbf{Z}^{\text{DD}}_{\hat{n}\hat{n}} | \mathbf{I}_{{\hat{n}}\setminus{\hat{m}}} \rrbracket \big)$, which is computationally accessible.

The resource quantifiers at temporal resolution $\hat{m}$, $\overline{I}_{\hat{m}}, \overline{M}_{\hat{m}}$ and $\overline{N}_{\hat{m}}$, also shed light on the trade-off between resources during dynamical decoupling. Due to their monotonicity, $\overline{I}_{\hat{m}},\overline{M}_{\hat{m}},\overline{N}_{\hat{m}}$ can only decrease under the allowed transformations of IQI. However -- by definition -- for each of them there is (at least) one respective transformation that will not decrease its value. We denote the corresponding superprocesses by $\mathbf{Z}^{I}_{\hat{n}\hat{n}},\mathbf{Z}^{M}_{\hat{n}\hat{n}},\mathbf{Z}^{N}_{\hat{n}\hat{n}}$, yielding the optimal values for $\overline{I}_{\hat{m}},\overline{M}_{\hat{m}}$ and $\overline{N}_{\hat{m}}$, respectively, given temporal resolution $\hat{m}$. Critically, these superprocesses need not be the same for different quantifiers, and hence optimising $\overline{I}_{\hat{m}}$ or $\overline{M}_{\hat{m}}$ in general requires sacrificing $\overline{N}_{\hat{m}}$. That is, a control sequence acting on $\mathbf{T}_{\hat{n}}$ that yields the highest overall mutual information at temporal resolution $\hat{m}$ will generally lead to a process $\llbracket \mathbf{T}_{\hat{n}}|\mathbf{Z}^{I}_{\hat{n}\hat{n}}|\mathbf{I}_{\hat{n}\setminus \hat{m}}\rrbracket$ whose non-Markovianity $\overline{N}_{\hat{m}}$ is lower than what could have potentially been obtained from $\mathbf{T}_{\hat{n}}$.  It is in this sense that non-Markovianity is `expended' during DD. 

Furthermore, applying $\mathbf{Z}^{I}_{\hat{n}\hat{n}}$ followed by coarse-graining to $\hat{m}$ does -- by definition -- not decrease the value of $\overline{I}_{\hat{m}}$. Then we can understand the transformation
\begin{equation}
    \mathbf{T}_{\hat{n}}   \mapsto \mathbf{T}'_{\hat{m}} = \llbracket \mathbf{T}_{\hat{n}} | \mathbf{Z}^{I}_{\hat{n}\hat{n}} | \mathbf{I}_{{\hat{n}}\setminus{\hat{m}}} \rrbracket  ,
\end{equation}
as lossless resource distillation in time, where $\overline{I}_{\hat{m}}(\mathbf{T}_{\hat{n}})=\overline{I}_{\hat{m}}(\mathbf{T}'_{\hat{m}})$ but $\hat{m} \subseteq \hat{n}$. This situation is analogous to resource distillation among spatially distinct subsystems, which is the basis of quantum error correction. Here, the resources contained in many temporally distinct subsystems of a `noisy' multitime process are concentrated into fewer `noisless' subsystems, e.g. an identity channel. Resource value is not increased \emph{overall}, yet a fine-grained noisy process can be converted to a coarse-grained noiseless one. A less optimal choice of superprocess than $\mathbf{Z}^{I}_{\hat{n}\hat{n}}$ will result in `lossy' resource distillation in time, or potentially even dilution. This is generically what happens if the experimenter chooses not to apply any kind of DD superprocess at all: resources present in the original fine-grained process are irreversibly wasted, and a resultant coarse-grained process/channel is noisier than it could have been if an adequate choice of control operation had been applied to it.

\section{Discussion}
Quantifying the amenability of quantum systems to control problems under specified experimental constraints is critical to the development of quantum technologies and the optimal usage of available resources. However, it is more challenging than quantifying the value of quantum states, particularly in the multitime paradigm~\cite{zambon2024processtensordistinguishabilitymeasures}. Here, we presented a generic methodology to construct operationally motivated resource quantifiers, focusing on the problem of noise reduction. This yielded quantifiers $\overline{I}_{\hat{m}},\overline{M}_{\hat{m}},\overline{N}_{\hat{m}}$ for total, Markovian, and non-Markovian correlations at temporal resolution $\hat{m}$, tailored to the respective sets $\mathsf{Z}_{\hat{n}\hat{m}}$ of available superprocesses. 

We found that each of the quantifiers is monotonic under the application of superprocesses in $\mathsf{Z}_{\hat{n}\hat{m}}$, making them well-behaved resource quantifiers. In addition, we demonstrated that $\overline{I}_{\hat{m}},\overline{M}_{\hat{m}}$ and $\overline{N}_{\hat{m}}$ are not mutually independent but satisfy a subadditivity property due to the fact that each of them is achieved by a different superprocess. As a consequence, optimising control operations to maximise one type of correlations, say, the total information $\overline{I}$, necessarily `wastes' the other two. 

Importantly, the monotonic behaviour of $\overline{I}_{\hat{m}},\overline{M}_{\hat{m}}$ and $\overline{N}_{\hat{m}}$ under control operations is limited only to superprocesses in $\mathsf{Z}_{\hat{n}\hat{m}}$. This is in contrast to generalised comb divergences $D$, a notion of distance between multi-time quantum processes which is monotonic under the application of \textit{arbitrary} superprocesses~\cite{zambon2024processtensordistinguishabilitymeasures}. We linked this concept with the resource quantifiers we introduced to this via \textit{reachable} comb divergences $D^\text{reach}$, which we showed to also be monotonic under $\mathsf{Z}_{\hat{n}\hat{m}}$. While we establish the hierarchy $\overline{I}_{\hat{m}} \leq D^{\text{reach}}_{\hat{m}}(\mathbf{T}_{\hat{n}}|| \mathbf{T}^{\text{marg}}_{\hat{n}})\leq D_{\hat{m}}(\mathbf{T}_{\hat{n}}|| \mathbf{T}^{\text{marg}}_{\hat{n}})$ (and similarly for $\overline{M}_{\hat{m}}$ and $\overline{N}_{\hat{m}}$), it is unknown whether there are cases of \textit{strict} inequality between $\overline{I}_{\hat{m}},\overline{M}_{\hat{m}},\overline{N}_{\hat{m}}$ and $D^\text{reach}$. Additionally, the existence, potential causes, and implications of strict separation for $D^\text{reach}$ and $D$ in this hierarchy remain to be investigated. When this hierarchy does not collapse, what resource is being captured by one quantifier but not the other?

Our results provide an avenue forward to faithfully quantify the resources of complex quantum processes and devices, given the practical control available to an experimenter. Notably, the case we investigated in most detail here is based on the assumption that the experimenter has no adaptive memory between timesteps. This is reasonable for the NISQ era, where adaptive control based on mid-circuit measurements can be difficult. However, our results immediately apply to more general cases, and imbuing the experimenter with (potentially restricted) memory may yield quantifiers that more accurately reflect the ultimate utility of a quantum system for a control task.

Going beyond the single-copy paradigm, we also investigated how the resource quantifiers we introduce behave under parallel and sequential composition, making them applicable to resource distillation setups when multiple copies of a given quantum process are available and can be processed. Here, in particular, superadditivity under parallel composition, as shown in Prop.~\ref{prop:superadditivity} suggests that regularisations of the quantifiers we present may be required to properly assess the amenability of composite quantum systems to noise reduction, or more general quantum control. We expect such treatment to extend the quantifiable utility of quantum processes beyond either the single-shot multitime or single-time asymptotic limits.

\section{Acknowledgements}
We thank Guilherme Zambon for bringing the error in Ref.~\cite{extractingdynamicalquantumresources} to our attention, and for fruitful discussions on this topic.
SM acknowledges funding from the European Union's Horizon Europe research and innovation programme under the Marie Sk{\l}odowska-Curie grant agreement No.\ 101068332.

\bibliographystyle{apsrev4-1}
\bibliography{refs.bib}

\begin{thebibliography}{51}%
\makeatletter
\providecommand \@ifxundefined [1]{%
 \@ifx{#1\undefined}
}%
\providecommand \@ifnum [1]{%
 \ifnum #1\expandafter \@firstoftwo
 \else \expandafter \@secondoftwo
 \fi
}%
\providecommand \@ifx [1]{%
 \ifx #1\expandafter \@firstoftwo
 \else \expandafter \@secondoftwo
 \fi
}%
\providecommand \natexlab [1]{#1}%
\providecommand \enquote  [1]{``#1''}%
\providecommand \bibnamefont  [1]{#1}%
\providecommand \bibfnamefont [1]{#1}%
\providecommand \citenamefont [1]{#1}%
\providecommand \href@noop [0]{\@secondoftwo}%
\providecommand \href [0]{\begingroup \@sanitize@url \@href}%
\providecommand \@href[1]{\@@startlink{#1}\@@href}%
\providecommand \@@href[1]{\endgroup#1\@@endlink}%
\providecommand \@sanitize@url [0]{\catcode `\\12\catcode `\$12\catcode `\&12\catcode `\#12\catcode `\^12\catcode `\_12\catcode `\%12\relax}%
\providecommand \@@startlink[1]{}%
\providecommand \@@endlink[0]{}%
\providecommand \url  [0]{\begingroup\@sanitize@url \@url }%
\providecommand \@url [1]{\endgroup\@href {#1}{\urlprefix }}%
\providecommand \urlprefix  [0]{URL }%
\providecommand \Eprint [0]{\href }%
\providecommand \doibase [0]{http://dx.doi.org/}%
\providecommand \selectlanguage [0]{\@gobble}%
\providecommand \bibinfo  [0]{\@secondoftwo}%
\providecommand \bibfield  [0]{\@secondoftwo}%
\providecommand \translation [1]{[#1]}%
\providecommand \BibitemOpen [0]{}%
\providecommand \bibitemStop [0]{}%
\providecommand \bibitemNoStop [0]{.\EOS\space}%
\providecommand \EOS [0]{\spacefactor3000\relax}%
\providecommand \BibitemShut  [1]{\csname bibitem#1\endcsname}%
\let\auto@bib@innerbib\@empty
\bibitem [{\citenamefont {Berk}\ \emph {et~al.}(2023)\citenamefont {Berk}, \citenamefont {Milz}, \citenamefont {Pollock},\ and\ \citenamefont {Modi}}]{extractingdynamicalquantumresources}%
  \BibitemOpen
  \bibfield  {author} {\bibinfo {author} {\bibfnamefont {G.~D.}\ \bibnamefont {Berk}}, \bibinfo {author} {\bibfnamefont {S.}~\bibnamefont {Milz}}, \bibinfo {author} {\bibfnamefont {F.~A.}\ \bibnamefont {Pollock}}, \ and\ \bibinfo {author} {\bibfnamefont {K.}~\bibnamefont {Modi}},\ }\href {\doibase 10.1038/s41534-023-00774-w} {\bibfield  {journal} {\bibinfo  {journal} {npj Quantum Information}\ }\textbf {\bibinfo {volume} {9}},\ \bibinfo {pages} {104} (\bibinfo {year} {2023})}\BibitemShut {NoStop}%
\bibitem [{\citenamefont {Shor}(1995)}]{PhysRevA.52.R2493}%
  \BibitemOpen
  \bibfield  {author} {\bibinfo {author} {\bibfnamefont {P.~W.}\ \bibnamefont {Shor}},\ }\href {\doibase 10.1103/PhysRevA.52.R2493} {\bibfield  {journal} {\bibinfo  {journal} {Phys. Rev. A}\ }\textbf {\bibinfo {volume} {52}},\ \bibinfo {pages} {R2493} (\bibinfo {year} {1995})}\BibitemShut {NoStop}%
\bibitem [{\citenamefont {Terhal}(2015)}]{RevModPhys.87.307}%
  \BibitemOpen
  \bibfield  {author} {\bibinfo {author} {\bibfnamefont {B.~M.}\ \bibnamefont {Terhal}},\ }\href {\doibase 10.1103/RevModPhys.87.307} {\bibfield  {journal} {\bibinfo  {journal} {Rev. Mod. Phys.}\ }\textbf {\bibinfo {volume} {87}},\ \bibinfo {pages} {307} (\bibinfo {year} {2015})}\BibitemShut {NoStop}%
\bibitem [{\citenamefont {Knill}(2005)}]{Knill_2005}%
  \BibitemOpen
  \bibfield  {author} {\bibinfo {author} {\bibfnamefont {E.}~\bibnamefont {Knill}},\ }\href {\doibase 10.1038/nature03350} {\bibfield  {journal} {\bibinfo  {journal} {Nature}\ }\textbf {\bibinfo {volume} {434}},\ \bibinfo {pages} {39–44} (\bibinfo {year} {2005})}\BibitemShut {NoStop}%
\bibitem [{\citenamefont {Temme}\ \emph {et~al.}(2017)\citenamefont {Temme}, \citenamefont {Bravyi},\ and\ \citenamefont {Gambetta}}]{PhysRevLett.119.180509}%
  \BibitemOpen
  \bibfield  {author} {\bibinfo {author} {\bibfnamefont {K.}~\bibnamefont {Temme}}, \bibinfo {author} {\bibfnamefont {S.}~\bibnamefont {Bravyi}}, \ and\ \bibinfo {author} {\bibfnamefont {J.~M.}\ \bibnamefont {Gambetta}},\ }\href {\doibase 10.1103/PhysRevLett.119.180509} {\bibfield  {journal} {\bibinfo  {journal} {Phys. Rev. Lett.}\ }\textbf {\bibinfo {volume} {119}},\ \bibinfo {pages} {180509} (\bibinfo {year} {2017})}\BibitemShut {NoStop}%
\bibitem [{\citenamefont {Endo}\ \emph {et~al.}(2018)\citenamefont {Endo}, \citenamefont {Benjamin},\ and\ \citenamefont {Li}}]{PhysRevX.8.031027}%
  \BibitemOpen
  \bibfield  {author} {\bibinfo {author} {\bibfnamefont {S.}~\bibnamefont {Endo}}, \bibinfo {author} {\bibfnamefont {S.~C.}\ \bibnamefont {Benjamin}}, \ and\ \bibinfo {author} {\bibfnamefont {Y.}~\bibnamefont {Li}},\ }\href {\doibase 10.1103/PhysRevX.8.031027} {\bibfield  {journal} {\bibinfo  {journal} {Phys. Rev. X}\ }\textbf {\bibinfo {volume} {8}},\ \bibinfo {pages} {031027} (\bibinfo {year} {2018})}\BibitemShut {NoStop}%
\bibitem [{\citenamefont {Kandala}\ \emph {et~al.}(2019)\citenamefont {Kandala}, \citenamefont {Temme}, \citenamefont {Córcoles}, \citenamefont {Mezzacapo}, \citenamefont {Chow},\ and\ \citenamefont {Gambetta}}]{Kandala_2019}%
  \BibitemOpen
  \bibfield  {author} {\bibinfo {author} {\bibfnamefont {A.}~\bibnamefont {Kandala}}, \bibinfo {author} {\bibfnamefont {K.}~\bibnamefont {Temme}}, \bibinfo {author} {\bibfnamefont {A.~D.}\ \bibnamefont {Córcoles}}, \bibinfo {author} {\bibfnamefont {A.}~\bibnamefont {Mezzacapo}}, \bibinfo {author} {\bibfnamefont {J.~M.}\ \bibnamefont {Chow}}, \ and\ \bibinfo {author} {\bibfnamefont {J.~M.}\ \bibnamefont {Gambetta}},\ }\href {\doibase 10.1038/s41586-019-1040-7} {\bibfield  {journal} {\bibinfo  {journal} {Nature}\ }\textbf {\bibinfo {volume} {567}},\ \bibinfo {pages} {491–495} (\bibinfo {year} {2019})}\BibitemShut {NoStop}%
\bibitem [{\citenamefont {Viola}\ \emph {et~al.}(1999)\citenamefont {Viola}, \citenamefont {Knill},\ and\ \citenamefont {Lloyd}}]{dynamicaldecouplingofopenquantumsystems}%
  \BibitemOpen
  \bibfield  {author} {\bibinfo {author} {\bibfnamefont {L.}~\bibnamefont {Viola}}, \bibinfo {author} {\bibfnamefont {E.}~\bibnamefont {Knill}}, \ and\ \bibinfo {author} {\bibfnamefont {S.}~\bibnamefont {Lloyd}},\ }\href {https://link.aps.org/doi/10.1103/PhysRevLett.82.2417} {\bibfield  {journal} {\bibinfo  {journal} {Phys. Rev. Lett.}\ }\textbf {\bibinfo {volume} {82}},\ \bibinfo {pages} {2417} (\bibinfo {year} {1999})}\BibitemShut {NoStop}%
\bibitem [{\citenamefont {Pokharel}\ \emph {et~al.}(2018)\citenamefont {Pokharel}, \citenamefont {Anand}, \citenamefont {Fortman},\ and\ \citenamefont {Lidar}}]{demonstrationoffidelityimprovement}%
  \BibitemOpen
  \bibfield  {author} {\bibinfo {author} {\bibfnamefont {B.}~\bibnamefont {Pokharel}}, \bibinfo {author} {\bibfnamefont {N.}~\bibnamefont {Anand}}, \bibinfo {author} {\bibfnamefont {B.}~\bibnamefont {Fortman}}, \ and\ \bibinfo {author} {\bibfnamefont {D.~A.}\ \bibnamefont {Lidar}},\ }\href {\doibase 10.1103/PhysRevLett.121.220502} {\bibfield  {journal} {\bibinfo  {journal} {Phys. Rev. Lett.}\ }\textbf {\bibinfo {volume} {121}},\ \bibinfo {pages} {220502} (\bibinfo {year} {2018})}\BibitemShut {NoStop}%
\bibitem [{\citenamefont {Naydenov}\ \emph {et~al.}(2011)\citenamefont {Naydenov}, \citenamefont {Dolde}, \citenamefont {Hall}, \citenamefont {Shin}, \citenamefont {Fedder}, \citenamefont {Hollenberg}, \citenamefont {Jelezko},\ and\ \citenamefont {Wrachtrup}}]{dynamicaldecouplingofasingleelectronspinatroom}%
  \BibitemOpen
  \bibfield  {author} {\bibinfo {author} {\bibfnamefont {B.}~\bibnamefont {Naydenov}}, \bibinfo {author} {\bibfnamefont {F.}~\bibnamefont {Dolde}}, \bibinfo {author} {\bibfnamefont {L.~T.}\ \bibnamefont {Hall}}, \bibinfo {author} {\bibfnamefont {C.}~\bibnamefont {Shin}}, \bibinfo {author} {\bibfnamefont {H.}~\bibnamefont {Fedder}}, \bibinfo {author} {\bibfnamefont {L.~C.~L.}\ \bibnamefont {Hollenberg}}, \bibinfo {author} {\bibfnamefont {F.}~\bibnamefont {Jelezko}}, \ and\ \bibinfo {author} {\bibfnamefont {J.}~\bibnamefont {Wrachtrup}},\ }\href {\doibase 10.1103/PhysRevB.83.081201} {\bibfield  {journal} {\bibinfo  {journal} {Phys. Rev. B}\ }\textbf {\bibinfo {volume} {83}},\ \bibinfo {pages} {081201} (\bibinfo {year} {2011})}\BibitemShut {NoStop}%
\bibitem [{\citenamefont {Du}\ \emph {et~al.}(2009)\citenamefont {Du}, \citenamefont {Rong}, \citenamefont {Zhao}, \citenamefont {Wang}, \citenamefont {Yang},\ and\ \citenamefont {Liu}}]{preservingelectronspincoherence}%
  \BibitemOpen
  \bibfield  {author} {\bibinfo {author} {\bibfnamefont {J.}~\bibnamefont {Du}}, \bibinfo {author} {\bibfnamefont {X.}~\bibnamefont {Rong}}, \bibinfo {author} {\bibfnamefont {N.}~\bibnamefont {Zhao}}, \bibinfo {author} {\bibfnamefont {Y.}~\bibnamefont {Wang}}, \bibinfo {author} {\bibfnamefont {J.}~\bibnamefont {Yang}}, \ and\ \bibinfo {author} {\bibfnamefont {R.~B.}\ \bibnamefont {Liu}},\ }\href {\doibase 10.1038/nature08470} {\bibfield  {journal} {\bibinfo  {journal} {Nature}\ }\textbf {\bibinfo {volume} {461}},\ \bibinfo {pages} {1265} (\bibinfo {year} {2009})}\BibitemShut {NoStop}%
\bibitem [{\citenamefont {Tong}\ \emph {et~al.}(2025)\citenamefont {Tong}, \citenamefont {Zhang},\ and\ \citenamefont {Pokharel}}]{tong2025empiricallearningdynamicaldecoupling}%
  \BibitemOpen
  \bibfield  {author} {\bibinfo {author} {\bibfnamefont {C.}~\bibnamefont {Tong}}, \bibinfo {author} {\bibfnamefont {H.}~\bibnamefont {Zhang}}, \ and\ \bibinfo {author} {\bibfnamefont {B.}~\bibnamefont {Pokharel}},\ }\href {\doibase 10.1103/h7pq-s159} {\bibfield  {journal} {\bibinfo  {journal} {PRX Quantum}\ }\textbf {\bibinfo {volume} {6}},\ \bibinfo {pages} {030319} (\bibinfo {year} {2025})}\BibitemShut {NoStop}%
\bibitem [{\citenamefont {Tripathi}\ \emph {et~al.}(2025)\citenamefont {Tripathi}, \citenamefont {Goss}, \citenamefont {Vezvaee}, \citenamefont {Nguyen}, \citenamefont {Siddiqi},\ and\ \citenamefont {Lidar}}]{PhysRevLett.134.050601}%
  \BibitemOpen
  \bibfield  {author} {\bibinfo {author} {\bibfnamefont {V.}~\bibnamefont {Tripathi}}, \bibinfo {author} {\bibfnamefont {N.}~\bibnamefont {Goss}}, \bibinfo {author} {\bibfnamefont {A.}~\bibnamefont {Vezvaee}}, \bibinfo {author} {\bibfnamefont {L.~B.}\ \bibnamefont {Nguyen}}, \bibinfo {author} {\bibfnamefont {I.}~\bibnamefont {Siddiqi}}, \ and\ \bibinfo {author} {\bibfnamefont {D.~A.}\ \bibnamefont {Lidar}},\ }\href {\doibase 10.1103/PhysRevLett.134.050601} {\bibfield  {journal} {\bibinfo  {journal} {Phys. Rev. Lett.}\ }\textbf {\bibinfo {volume} {134}},\ \bibinfo {pages} {050601} (\bibinfo {year} {2025})}\BibitemShut {NoStop}%
\bibitem [{\citenamefont {Ezzell}\ \emph {et~al.}(2023)\citenamefont {Ezzell}, \citenamefont {Pokharel}, \citenamefont {Tewala}, \citenamefont {Quiroz},\ and\ \citenamefont {Lidar}}]{PhysRevApplied.20.064027}%
  \BibitemOpen
  \bibfield  {author} {\bibinfo {author} {\bibfnamefont {N.}~\bibnamefont {Ezzell}}, \bibinfo {author} {\bibfnamefont {B.}~\bibnamefont {Pokharel}}, \bibinfo {author} {\bibfnamefont {L.}~\bibnamefont {Tewala}}, \bibinfo {author} {\bibfnamefont {G.}~\bibnamefont {Quiroz}}, \ and\ \bibinfo {author} {\bibfnamefont {D.~A.}\ \bibnamefont {Lidar}},\ }\href {\doibase 10.1103/PhysRevApplied.20.064027} {\bibfield  {journal} {\bibinfo  {journal} {Phys. Rev. Appl.}\ }\textbf {\bibinfo {volume} {20}},\ \bibinfo {pages} {064027} (\bibinfo {year} {2023})}\BibitemShut {NoStop}%
\bibitem [{\citenamefont {Breuer}\ and\ \citenamefont {Petruccione}(2007)}]{Breuer:2007juk}%
  \BibitemOpen
  \bibfield  {author} {\bibinfo {author} {\bibfnamefont {H.-P.}\ \bibnamefont {Breuer}}\ and\ \bibinfo {author} {\bibfnamefont {F.}~\bibnamefont {Petruccione}},\ }\href {\doibase 10.1093/acprof:oso/9780199213900.001.0001} {\emph {\bibinfo {title} {{The Theory of Open Quantum Systems}}}}\ (\bibinfo  {publisher} {Oxford University Press},\ \bibinfo {year} {2007})\BibitemShut {NoStop}%
\bibitem [{\citenamefont {Pollock}\ \emph {et~al.}(2018)\citenamefont {Pollock}, \citenamefont {Rodr{\'{\i}}guez-Rosario}, \citenamefont {Frauenheim}, \citenamefont {Paternostro},\ and\ \citenamefont {Modi}}]{nonmarkovianquantumprocesses}%
  \BibitemOpen
  \bibfield  {author} {\bibinfo {author} {\bibfnamefont {F.~A.}\ \bibnamefont {Pollock}}, \bibinfo {author} {\bibfnamefont {C.}~\bibnamefont {Rodr{\'{\i}}guez-Rosario}}, \bibinfo {author} {\bibfnamefont {T.}~\bibnamefont {Frauenheim}}, \bibinfo {author} {\bibfnamefont {M.}~\bibnamefont {Paternostro}}, \ and\ \bibinfo {author} {\bibfnamefont {K.}~\bibnamefont {Modi}},\ }\href {\doibase 10.1103/PhysRevA.97.012127} {\bibfield  {journal} {\bibinfo  {journal} {Phys. Rev. A}\ }\textbf {\bibinfo {volume} {97}},\ \bibinfo {pages} {012127} (\bibinfo {year} {2018})}\BibitemShut {NoStop}%
\bibitem [{\citenamefont {Milz}\ and\ \citenamefont {Modi}(2021)}]{quantumstochasticprocesses}%
  \BibitemOpen
  \bibfield  {author} {\bibinfo {author} {\bibfnamefont {S.}~\bibnamefont {Milz}}\ and\ \bibinfo {author} {\bibfnamefont {K.}~\bibnamefont {Modi}},\ }\href {\doibase 10.1103/PRXQuantum.2.030201} {\bibfield  {journal} {\bibinfo  {journal} {PRX Quantum}\ }\textbf {\bibinfo {volume} {2}},\ \bibinfo {pages} {030201} (\bibinfo {year} {2021})}\BibitemShut {NoStop}%
\bibitem [{\citenamefont {Crutchfield}\ and\ \citenamefont {Feldman}(2001)}]{regularitiesunseenrandomnessobserved}%
  \BibitemOpen
  \bibfield  {author} {\bibinfo {author} {\bibfnamefont {J.~P.}\ \bibnamefont {Crutchfield}}\ and\ \bibinfo {author} {\bibfnamefont {D.~P.}\ \bibnamefont {Feldman}},\ }\href {https://arxiv.org/abs/cond-mat/0102181} {\bibfield  {journal} {\bibinfo  {journal} {arXiv:cond-mat/0102181}\ } (\bibinfo {year} {2001})}\BibitemShut {NoStop}%
\bibitem [{\citenamefont {Pechukas}(1994)}]{reduceddynamicsneednotbe}%
  \BibitemOpen
  \bibfield  {author} {\bibinfo {author} {\bibfnamefont {P.}~\bibnamefont {Pechukas}},\ }\href {\doibase 10.1103/PhysRevLett.73.1060} {\bibfield  {journal} {\bibinfo  {journal} {Phys. Rev. Lett.}\ }\textbf {\bibinfo {volume} {73}},\ \bibinfo {pages} {1060} (\bibinfo {year} {1994})}\BibitemShut {NoStop}%
\bibitem [{\citenamefont {Morris}\ \emph {et~al.}(2022)\citenamefont {Morris}, \citenamefont {Pollock},\ and\ \citenamefont {Modi}}]{quantifyingnonmarkovianityinasuperconditing}%
  \BibitemOpen
  \bibfield  {author} {\bibinfo {author} {\bibfnamefont {J.}~\bibnamefont {Morris}}, \bibinfo {author} {\bibfnamefont {F.~A.}\ \bibnamefont {Pollock}}, \ and\ \bibinfo {author} {\bibfnamefont {K.}~\bibnamefont {Modi}},\ }\href {http://dx.doi.org/10.1142/S123016122250007X} {\bibfield  {journal} {\bibinfo  {journal} {Open Sys. Inf. Dyn.}\ }\textbf {\bibinfo {volume} {29}} (\bibinfo {year} {2022})}\BibitemShut {NoStop}%
\bibitem [{\citenamefont {White}\ \emph {et~al.}(2020)\citenamefont {White}, \citenamefont {Hill}, \citenamefont {Pollock}, \citenamefont {Hollenberg},\ and\ \citenamefont {Modi}}]{White_2020}%
  \BibitemOpen
  \bibfield  {author} {\bibinfo {author} {\bibfnamefont {G.~A.~L.}\ \bibnamefont {White}}, \bibinfo {author} {\bibfnamefont {C.~D.}\ \bibnamefont {Hill}}, \bibinfo {author} {\bibfnamefont {F.~A.}\ \bibnamefont {Pollock}}, \bibinfo {author} {\bibfnamefont {L.~C.~L.}\ \bibnamefont {Hollenberg}}, \ and\ \bibinfo {author} {\bibfnamefont {K.}~\bibnamefont {Modi}},\ }\href {http://dx.doi.org/10.1038/s41467-020-20113-3} {\bibfield  {journal} {\bibinfo  {journal} {Nat. Commun.}\ }\textbf {\bibinfo {volume} {11}} (\bibinfo {year} {2020})}\BibitemShut {NoStop}%
\bibitem [{\citenamefont {White}\ \emph {et~al.}(2025)\citenamefont {White}, \citenamefont {Jurcevic}, \citenamefont {Hill},\ and\ \citenamefont {Modi}}]{unifyingnonmarkoviancharacterisationefficient}%
  \BibitemOpen
  \bibfield  {author} {\bibinfo {author} {\bibfnamefont {G.~A.~L.}\ \bibnamefont {White}}, \bibinfo {author} {\bibfnamefont {P.}~\bibnamefont {Jurcevic}}, \bibinfo {author} {\bibfnamefont {C.~D.}\ \bibnamefont {Hill}}, \ and\ \bibinfo {author} {\bibfnamefont {K.}~\bibnamefont {Modi}},\ }\href {\doibase 10.1103/PhysRevX.15.021047} {\bibfield  {journal} {\bibinfo  {journal} {Phys. Rev. X}\ }\textbf {\bibinfo {volume} {15}},\ \bibinfo {pages} {021047} (\bibinfo {year} {2025})}\BibitemShut {NoStop}%
\bibitem [{\citenamefont {Tanggara}\ \emph {et~al.}(2024)\citenamefont {Tanggara}, \citenamefont {Gu},\ and\ \citenamefont {Bharti}}]{strategiccodeunifiedspatiotemporal}%
  \BibitemOpen
  \bibfield  {author} {\bibinfo {author} {\bibfnamefont {A.}~\bibnamefont {Tanggara}}, \bibinfo {author} {\bibfnamefont {M.}~\bibnamefont {Gu}}, \ and\ \bibinfo {author} {\bibfnamefont {K.}~\bibnamefont {Bharti}},\ }\href {https://arxiv.org/abs/2405.17567} {\bibfield  {journal} {\bibinfo  {journal} {arXiv:2405.17567}\ } (\bibinfo {year} {2024})}\BibitemShut {NoStop}%
\bibitem [{\citenamefont {Zambon}\ and\ \citenamefont {Adesso}(2025)}]{quantumprocessesthermodynamicresources}%
  \BibitemOpen
  \bibfield  {author} {\bibinfo {author} {\bibfnamefont {G.}~\bibnamefont {Zambon}}\ and\ \bibinfo {author} {\bibfnamefont {G.}~\bibnamefont {Adesso}},\ }\href {\doibase 10.1103/PhysRevLett.134.200401} {\bibfield  {journal} {\bibinfo  {journal} {Phys. Rev. Lett.}\ }\textbf {\bibinfo {volume} {134}},\ \bibinfo {pages} {200401} (\bibinfo {year} {2025})}\BibitemShut {NoStop}%
\bibitem [{\citenamefont {Milz}\ \emph {et~al.}(2017)\citenamefont {Milz}, \citenamefont {Pollock},\ and\ \citenamefont {Modi}}]{anintroductiontooperational}%
  \BibitemOpen
  \bibfield  {author} {\bibinfo {author} {\bibfnamefont {S.}~\bibnamefont {Milz}}, \bibinfo {author} {\bibfnamefont {F.~A.}\ \bibnamefont {Pollock}}, \ and\ \bibinfo {author} {\bibfnamefont {K.}~\bibnamefont {Modi}},\ }\href {\doibase 10.1142/S1230161217400169} {\bibfield  {journal} {\bibinfo  {journal} {Open Sys. Info. Dyn.}\ }\textbf {\bibinfo {volume} {24}},\ \bibinfo {pages} {1740016} (\bibinfo {year} {2017})}\BibitemShut {NoStop}%
\bibitem [{\citenamefont {White}\ \emph {et~al.}(2022)\citenamefont {White}, \citenamefont {Pollock}, \citenamefont {Hollenberg}, \citenamefont {Modi},\ and\ \citenamefont {Hill}}]{PRXQuantum.3.020344}%
  \BibitemOpen
  \bibfield  {author} {\bibinfo {author} {\bibfnamefont {G.}~\bibnamefont {White}}, \bibinfo {author} {\bibfnamefont {F.}~\bibnamefont {Pollock}}, \bibinfo {author} {\bibfnamefont {L.}~\bibnamefont {Hollenberg}}, \bibinfo {author} {\bibfnamefont {K.}~\bibnamefont {Modi}}, \ and\ \bibinfo {author} {\bibfnamefont {C.}~\bibnamefont {Hill}},\ }\href {\doibase 10.1103/PRXQuantum.3.020344} {\bibfield  {journal} {\bibinfo  {journal} {PRX Quantum}\ }\textbf {\bibinfo {volume} {3}},\ \bibinfo {pages} {020344} (\bibinfo {year} {2022})}\BibitemShut {NoStop}%
\bibitem [{\citenamefont {Chitambar}\ and\ \citenamefont {Gour}(2019)}]{review}%
  \BibitemOpen
  \bibfield  {author} {\bibinfo {author} {\bibfnamefont {E.}~\bibnamefont {Chitambar}}\ and\ \bibinfo {author} {\bibfnamefont {G.}~\bibnamefont {Gour}},\ }\href {\doibase 10.1103/RevModPhys.91.025001} {\bibfield  {journal} {\bibinfo  {journal} {Rev. Mod. Phys.}\ }\textbf {\bibinfo {volume} {91}},\ \bibinfo {pages} {025001} (\bibinfo {year} {2019})}\BibitemShut {NoStop}%
\bibitem [{\citenamefont {Berta}\ \emph {et~al.}(2023)\citenamefont {Berta}, \citenamefont {Brandão}, \citenamefont {Gour}, \citenamefont {Lami}, \citenamefont {Plenio}, \citenamefont {Regula},\ and\ \citenamefont {Tomamichel}}]{Berta_2023}%
  \BibitemOpen
  \bibfield  {author} {\bibinfo {author} {\bibfnamefont {M.}~\bibnamefont {Berta}}, \bibinfo {author} {\bibfnamefont {F.~G. S.~L.}\ \bibnamefont {Brandão}}, \bibinfo {author} {\bibfnamefont {G.}~\bibnamefont {Gour}}, \bibinfo {author} {\bibfnamefont {L.}~\bibnamefont {Lami}}, \bibinfo {author} {\bibfnamefont {M.~B.}\ \bibnamefont {Plenio}}, \bibinfo {author} {\bibfnamefont {B.}~\bibnamefont {Regula}}, \ and\ \bibinfo {author} {\bibfnamefont {M.}~\bibnamefont {Tomamichel}},\ }\href {\doibase 10.22331/q-2023-09-07-1103} {\bibfield  {journal} {\bibinfo  {journal} {Quantum}\ }\textbf {\bibinfo {volume} {7}},\ \bibinfo {pages} {1103} (\bibinfo {year} {2023})}\BibitemShut {NoStop}%
\bibitem [{\citenamefont {Hayashi}\ and\ \citenamefont {Yamasaki}(2025)}]{hayashi2025generalizedquantumsteinslemma}%
  \BibitemOpen
  \bibfield  {author} {\bibinfo {author} {\bibfnamefont {M.}~\bibnamefont {Hayashi}}\ and\ \bibinfo {author} {\bibfnamefont {H.}~\bibnamefont {Yamasaki}},\ }\href {https://arxiv.org/abs/2408.02722} {\bibfield  {journal} {\bibinfo  {journal} {arXiv:2408.02722}\ } (\bibinfo {year} {2025})}\BibitemShut {NoStop}%
\bibitem [{\citenamefont {Burgarth}\ \emph {et~al.}(2021)\citenamefont {Burgarth}, \citenamefont {Facchi}, \citenamefont {Fraas},\ and\ \citenamefont {Hillier}}]{nonmarkoviannoisethatcannotbe}%
  \BibitemOpen
  \bibfield  {author} {\bibinfo {author} {\bibfnamefont {D.}~\bibnamefont {Burgarth}}, \bibinfo {author} {\bibfnamefont {P.}~\bibnamefont {Facchi}}, \bibinfo {author} {\bibfnamefont {M.}~\bibnamefont {Fraas}}, \ and\ \bibinfo {author} {\bibfnamefont {R.}~\bibnamefont {Hillier}},\ }\href {\doibase 10.21468/SciPostPhys.11.2.027} {\bibfield  {journal} {\bibinfo  {journal} {SciPost Phys.}\ }\textbf {\bibinfo {volume} {11}},\ \bibinfo {pages} {27} (\bibinfo {year} {2021})}\BibitemShut {NoStop}%
\bibitem [{\citenamefont {Addis}\ \emph {et~al.}(2015)\citenamefont {Addis}, \citenamefont {Ciccarello}, \citenamefont {Cascio}, \citenamefont {Palma},\ and\ \citenamefont {Maniscalco}}]{dynamicaldecouplingefficiencyversusquantumnonmarkovianity}%
  \BibitemOpen
  \bibfield  {author} {\bibinfo {author} {\bibfnamefont {C.}~\bibnamefont {Addis}}, \bibinfo {author} {\bibfnamefont {F.}~\bibnamefont {Ciccarello}}, \bibinfo {author} {\bibfnamefont {M.}~\bibnamefont {Cascio}}, \bibinfo {author} {\bibfnamefont {G.~M.}\ \bibnamefont {Palma}}, \ and\ \bibinfo {author} {\bibfnamefont {S.}~\bibnamefont {Maniscalco}},\ }\href {\doibase 10.1088/1367-2630/17/12/123004} {\bibfield  {journal} {\bibinfo  {journal} {New J. Phys}\ }\textbf {\bibinfo {volume} {17}},\ \bibinfo {pages} {123004} (\bibinfo {year} {2015})}\BibitemShut {NoStop}%
\bibitem [{\citenamefont {Gough}\ and\ \citenamefont {Nurdin}(2017)}]{Gough_2017}%
  \BibitemOpen
  \bibfield  {author} {\bibinfo {author} {\bibfnamefont {J.~E.}\ \bibnamefont {Gough}}\ and\ \bibinfo {author} {\bibfnamefont {H.~I.}\ \bibnamefont {Nurdin}},\ }in\ \href {\doibase 10.1109/cdc.2017.8264587} {\emph {\bibinfo {booktitle} {2017 IEEE 56th Annual Conference on Decision and Control (CDC)}}}\ (\bibinfo  {publisher} {IEEE},\ \bibinfo {year} {2017})\ p.\ \bibinfo {pages} {6155–6160}\BibitemShut {NoStop}%
\bibitem [{\citenamefont {Figueroa-Romero}\ \emph {et~al.}(2021)\citenamefont {Figueroa-Romero}, \citenamefont {Modi}, \citenamefont {Harris}, \citenamefont {Stace},\ and\ \citenamefont {Hsieh}}]{PRXQuantum.2.040351}%
  \BibitemOpen
  \bibfield  {author} {\bibinfo {author} {\bibfnamefont {P.}~\bibnamefont {Figueroa-Romero}}, \bibinfo {author} {\bibfnamefont {K.}~\bibnamefont {Modi}}, \bibinfo {author} {\bibfnamefont {R.~J.}\ \bibnamefont {Harris}}, \bibinfo {author} {\bibfnamefont {T.~M.}\ \bibnamefont {Stace}}, \ and\ \bibinfo {author} {\bibfnamefont {M.-H.}\ \bibnamefont {Hsieh}},\ }\href {\doibase 10.1103/PRXQuantum.2.040351} {\bibfield  {journal} {\bibinfo  {journal} {PRX Quantum}\ }\textbf {\bibinfo {volume} {2}},\ \bibinfo {pages} {040351} (\bibinfo {year} {2021})}\BibitemShut {NoStop}%
\bibitem [{\citenamefont {Figueroa-Romero}\ \emph {et~al.}(2024)\citenamefont {Figueroa-Romero}, \citenamefont {Papič}, \citenamefont {Auer}, \citenamefont {Hsieh}, \citenamefont {Modi},\ and\ \citenamefont {de~Vega}}]{Figueroa-Romero_2024}%
  \BibitemOpen
  \bibfield  {author} {\bibinfo {author} {\bibfnamefont {P.}~\bibnamefont {Figueroa-Romero}}, \bibinfo {author} {\bibfnamefont {M.}~\bibnamefont {Papič}}, \bibinfo {author} {\bibfnamefont {A.}~\bibnamefont {Auer}}, \bibinfo {author} {\bibfnamefont {M.-H.}\ \bibnamefont {Hsieh}}, \bibinfo {author} {\bibfnamefont {K.}~\bibnamefont {Modi}}, \ and\ \bibinfo {author} {\bibfnamefont {I.}~\bibnamefont {de~Vega}},\ }\href {\doibase 10.1088/2058-9565/ad3f44} {\bibfield  {journal} {\bibinfo  {journal} {Quantum Science and Technology}\ }\textbf {\bibinfo {volume} {9}},\ \bibinfo {pages} {035020} (\bibinfo {year} {2024})}\BibitemShut {NoStop}%
\bibitem [{\citenamefont {Zambon}(2024)}]{zambon2024processtensordistinguishabilitymeasures}%
  \BibitemOpen
  \bibfield  {author} {\bibinfo {author} {\bibfnamefont {G.}~\bibnamefont {Zambon}},\ }\href {\doibase 10.1103/PhysRevA.110.042210} {\bibfield  {journal} {\bibinfo  {journal} {Phys. Rev. A}\ }\textbf {\bibinfo {volume} {110}},\ \bibinfo {pages} {042210} (\bibinfo {year} {2024})}\BibitemShut {NoStop}%
\bibitem [{\citenamefont {Ng}\ \emph {et~al.}(2011)\citenamefont {Ng}, \citenamefont {Lidar},\ and\ \citenamefont {Preskill}}]{Ng_2011}%
  \BibitemOpen
  \bibfield  {author} {\bibinfo {author} {\bibfnamefont {H.~K.}\ \bibnamefont {Ng}}, \bibinfo {author} {\bibfnamefont {D.~A.}\ \bibnamefont {Lidar}}, \ and\ \bibinfo {author} {\bibfnamefont {J.}~\bibnamefont {Preskill}},\ }\href {http://dx.doi.org/10.1103/PhysRevA.84.012305} {\bibfield  {journal} {\bibinfo  {journal} {Phys. Rev. A}\ }\textbf {\bibinfo {volume} {84}} (\bibinfo {year} {2011})}\BibitemShut {NoStop}%
\bibitem [{\citenamefont {Paz-Silva}\ and\ \citenamefont {Lidar}(2013)}]{optimallycombining}%
  \BibitemOpen
  \bibfield  {author} {\bibinfo {author} {\bibfnamefont {G.~A.}\ \bibnamefont {Paz-Silva}}\ and\ \bibinfo {author} {\bibfnamefont {D.~A.}\ \bibnamefont {Lidar}},\ }\href {\doibase 10.1038/srep01530} {\bibfield  {journal} {\bibinfo  {journal} {Sci. Rep.}\ }\textbf {\bibinfo {volume} {3}},\ \bibinfo {pages} {1530} (\bibinfo {year} {2013})}\BibitemShut {NoStop}%
\bibitem [{\citenamefont {Han}\ \emph {et~al.}(2025)\citenamefont {Han}, \citenamefont {Zhang}, \citenamefont {Xue}, \citenamefont {Yu},\ and\ \citenamefont {Long}}]{han2024protectinglogicalqubitsdynamical}%
  \BibitemOpen
  \bibfield  {author} {\bibinfo {author} {\bibfnamefont {J.-X.}\ \bibnamefont {Han}}, \bibinfo {author} {\bibfnamefont {J.}~\bibnamefont {Zhang}}, \bibinfo {author} {\bibfnamefont {G.-M.}\ \bibnamefont {Xue}}, \bibinfo {author} {\bibfnamefont {H.}~\bibnamefont {Yu}}, \ and\ \bibinfo {author} {\bibfnamefont {G.}~\bibnamefont {Long}},\ }\href {https://link.aps.org/doi/10.1103/lm8x-r48q} {\bibfield  {journal} {\bibinfo  {journal} {Phys. Rev. Appl.}\ }\textbf {\bibinfo {volume} {24}},\ \bibinfo {pages} {024003} (\bibinfo {year} {2025})}\BibitemShut {NoStop}%
\bibitem [{\citenamefont {Chiribella}\ \emph {et~al.}(2009)\citenamefont {Chiribella}, \citenamefont {D'Ariano},\ and\ \citenamefont {Perinotti}}]{theoreticalframeworkforquantumnetworks}%
  \BibitemOpen
  \bibfield  {author} {\bibinfo {author} {\bibfnamefont {G.}~\bibnamefont {Chiribella}}, \bibinfo {author} {\bibfnamefont {G.~M.}\ \bibnamefont {D'Ariano}}, \ and\ \bibinfo {author} {\bibfnamefont {P.}~\bibnamefont {Perinotti}},\ }\href {\doibase 10.1103/PhysRevA.80.022339} {\bibfield  {journal} {\bibinfo  {journal} {Phys. Rev. A}\ }\textbf {\bibinfo {volume} {80}},\ \bibinfo {pages} {022339} (\bibinfo {year} {2009})}\BibitemShut {NoStop}%
\bibitem [{Note1()}]{Note1}%
  \BibitemOpen
  \bibinfo {note} {Here, the definition of the link product differs from the original one in Ref.~\cite {theoreticalframeworkforquantumnetworks} by a factor $d_s^{|\protect \mathbf {T} \cap \protect \mathbf {Z}|}$. This stems from the fact that the Choi matrices we use are normalised to unit trace.}\BibitemShut {Stop}%
\bibitem [{\citenamefont {Berk}\ \emph {et~al.}(2021)\citenamefont {Berk}, \citenamefont {Garner}, \citenamefont {Yadin}, \citenamefont {Modi},\ and\ \citenamefont {Pollock}}]{resourcetheoriesofmultitime}%
  \BibitemOpen
  \bibfield  {author} {\bibinfo {author} {\bibfnamefont {G.~D.}\ \bibnamefont {Berk}}, \bibinfo {author} {\bibfnamefont {A.~J.~P.}\ \bibnamefont {Garner}}, \bibinfo {author} {\bibfnamefont {B.}~\bibnamefont {Yadin}}, \bibinfo {author} {\bibfnamefont {K.}~\bibnamefont {Modi}}, \ and\ \bibinfo {author} {\bibfnamefont {F.~A.}\ \bibnamefont {Pollock}},\ }\href {\doibase 10.22331/q-2021-04-20-435} {\bibfield  {journal} {\bibinfo  {journal} {{Quantum}}\ }\textbf {\bibinfo {volume} {5}},\ \bibinfo {pages} {435} (\bibinfo {year} {2021})}\BibitemShut {NoStop}%
\bibitem [{Note2()}]{Note2}%
  \BibitemOpen
  \bibinfo {note} {Here, all process tensors are normalised to unit trace, such that entropic measures are well-defined}\BibitemShut {NoStop}%
\bibitem [{\citenamefont {Capela}\ \emph {et~al.}(2020)\citenamefont {Capela}, \citenamefont {C{\'e}leri}, \citenamefont {Modi},\ and\ \citenamefont {Chaves}}]{capela_monogamy_2020}%
  \BibitemOpen
  \bibfield  {author} {\bibinfo {author} {\bibfnamefont {M.}~\bibnamefont {Capela}}, \bibinfo {author} {\bibfnamefont {L.~C.}\ \bibnamefont {C{\'e}leri}}, \bibinfo {author} {\bibfnamefont {K.}~\bibnamefont {Modi}}, \ and\ \bibinfo {author} {\bibfnamefont {R.}~\bibnamefont {Chaves}},\ }\href {\doibase 10.1103/PhysRevResearch.2.013350} {\bibfield  {journal} {\bibinfo  {journal} {Phys. Rev. Res.}\ }\textbf {\bibinfo {volume} {2}},\ \bibinfo {pages} {013350} (\bibinfo {year} {2020})}\BibitemShut {NoStop}%
\bibitem [{\citenamefont {Zambon}\ and\ \citenamefont {Soares-Pinto}(2024)}]{Zambon_2024}%
  \BibitemOpen
  \bibfield  {author} {\bibinfo {author} {\bibfnamefont {G.}~\bibnamefont {Zambon}}\ and\ \bibinfo {author} {\bibfnamefont {D.~O.}\ \bibnamefont {Soares-Pinto}},\ }\href {\doibase 10.1103/physreva.109.062401} {\bibfield  {journal} {\bibinfo  {journal} {Phys. Rev. A}\ }\textbf {\bibinfo {volume} {109}},\ \bibinfo {pages} {062401} (\bibinfo {year} {2024})}\BibitemShut {NoStop}%
\bibitem [{\citenamefont {Hirche}(2023)}]{Hirche2023}%
  \BibitemOpen
  \bibfield  {author} {\bibinfo {author} {\bibfnamefont {C.}~\bibnamefont {Hirche}},\ }\href {\doibase https://doi.org/10.22331/q-2023-07-25-1064} {\bibfield  {journal} {\bibinfo  {journal} {Quantum}\ }\textbf {\bibinfo {volume} {7}},\ \bibinfo {pages} {1064} (\bibinfo {year} {2023})}\BibitemShut {NoStop}%
\bibitem [{\citenamefont {Barry}\ \emph {et~al.}(2014)\citenamefont {Barry}, \citenamefont {Barry},\ and\ \citenamefont {Aaronson}}]{Barry_2014}%
  \BibitemOpen
  \bibfield  {author} {\bibinfo {author} {\bibfnamefont {J.}~\bibnamefont {Barry}}, \bibinfo {author} {\bibfnamefont {D.~T.}\ \bibnamefont {Barry}}, \ and\ \bibinfo {author} {\bibfnamefont {S.}~\bibnamefont {Aaronson}},\ }\href {http://dx.doi.org/10.1103/PhysRevA.90.032311} {\bibfield  {journal} {\bibinfo  {journal} {Phys. Rev. A}\ }\textbf {\bibinfo {volume} {90}} (\bibinfo {year} {2014})}\BibitemShut {NoStop}%
\bibitem [{\citenamefont {Ying}\ \emph {et~al.}(2021)\citenamefont {Ying}, \citenamefont {Feng},\ and\ \citenamefont {Ying}}]{Ying2021}%
  \BibitemOpen
  \bibfield  {author} {\bibinfo {author} {\bibfnamefont {M.-S.}\ \bibnamefont {Ying}}, \bibinfo {author} {\bibfnamefont {Y.}~\bibnamefont {Feng}}, \ and\ \bibinfo {author} {\bibfnamefont {S.-G.}\ \bibnamefont {Ying}},\ }\href {\doibase 10.1007/s11633-021-1278-z} {\bibfield  {journal} {\bibinfo  {journal} {Int. J. Autom. Comput.}\ }\textbf {\bibinfo {volume} {18}},\ \bibinfo {pages} {410} (\bibinfo {year} {2021})}\BibitemShut {NoStop}%
\bibitem [{\citenamefont {Werschnik}\ and\ \citenamefont {Gross}(2007)}]{werschnik2007quantumoptimalcontroltheory}%
  \BibitemOpen
  \bibfield  {author} {\bibinfo {author} {\bibfnamefont {J.}~\bibnamefont {Werschnik}}\ and\ \bibinfo {author} {\bibfnamefont {E.~K.~U.}\ \bibnamefont {Gross}},\ }\href {https://arxiv.org/abs/0707.1883} {\bibfield  {journal} {\bibinfo  {journal} {arXiv:0707.1883}\ } (\bibinfo {year} {2007})}\BibitemShut {NoStop}%
\bibitem [{\citenamefont {Koch}\ \emph {et~al.}(2022)\citenamefont {Koch}, \citenamefont {Boscain}, \citenamefont {Calarco}, \citenamefont {Dirr}, \citenamefont {Filipp}, \citenamefont {Glaser}, \citenamefont {Kosloff}, \citenamefont {Montangero}, \citenamefont {Schulte-Herbrüggen}, \citenamefont {Sugny},\ and\ \citenamefont {Wilhelm}}]{Koch_2022}%
  \BibitemOpen
  \bibfield  {author} {\bibinfo {author} {\bibfnamefont {C.~P.}\ \bibnamefont {Koch}}, \bibinfo {author} {\bibfnamefont {U.}~\bibnamefont {Boscain}}, \bibinfo {author} {\bibfnamefont {T.}~\bibnamefont {Calarco}}, \bibinfo {author} {\bibfnamefont {G.}~\bibnamefont {Dirr}}, \bibinfo {author} {\bibfnamefont {S.}~\bibnamefont {Filipp}}, \bibinfo {author} {\bibfnamefont {S.~J.}\ \bibnamefont {Glaser}}, \bibinfo {author} {\bibfnamefont {R.}~\bibnamefont {Kosloff}}, \bibinfo {author} {\bibfnamefont {S.}~\bibnamefont {Montangero}}, \bibinfo {author} {\bibfnamefont {T.}~\bibnamefont {Schulte-Herbrüggen}}, \bibinfo {author} {\bibfnamefont {D.}~\bibnamefont {Sugny}}, \ and\ \bibinfo {author} {\bibfnamefont {F.~K.}\ \bibnamefont {Wilhelm}},\ }\href {http://dx.doi.org/10.1140/epjqt/s40507-022-00138-x} {\bibfield  {journal} {\bibinfo  {journal} {EPJ Quantum Technology}\ }\textbf {\bibinfo {volume} {9}} (\bibinfo {year} {2022})}\BibitemShut {NoStop}%
\bibitem [{\citenamefont {Caneva}\ \emph {et~al.}(2011)\citenamefont {Caneva}, \citenamefont {Calarco},\ and\ \citenamefont {Montangero}}]{PhysRevA.84.022326}%
  \BibitemOpen
  \bibfield  {author} {\bibinfo {author} {\bibfnamefont {T.}~\bibnamefont {Caneva}}, \bibinfo {author} {\bibfnamefont {T.}~\bibnamefont {Calarco}}, \ and\ \bibinfo {author} {\bibfnamefont {S.}~\bibnamefont {Montangero}},\ }\href {\doibase 10.1103/PhysRevA.84.022326} {\bibfield  {journal} {\bibinfo  {journal} {Phys. Rev. A}\ }\textbf {\bibinfo {volume} {84}},\ \bibinfo {pages} {022326} (\bibinfo {year} {2011})}\BibitemShut {NoStop}%
\bibitem [{\citenamefont {Khaneja}\ \emph {et~al.}(2005)\citenamefont {Khaneja}, \citenamefont {Reiss}, \citenamefont {Kehlet}, \citenamefont {Schulte-Herbrüggen},\ and\ \citenamefont {Glaser}}]{KHANEJA2005296}%
  \BibitemOpen
  \bibfield  {author} {\bibinfo {author} {\bibfnamefont {N.}~\bibnamefont {Khaneja}}, \bibinfo {author} {\bibfnamefont {T.}~\bibnamefont {Reiss}}, \bibinfo {author} {\bibfnamefont {C.}~\bibnamefont {Kehlet}}, \bibinfo {author} {\bibfnamefont {T.}~\bibnamefont {Schulte-Herbrüggen}}, \ and\ \bibinfo {author} {\bibfnamefont {S.~J.}\ \bibnamefont {Glaser}},\ }\href {\doibase https://doi.org/10.1016/j.jmr.2004.11.004} {\bibfield  {journal} {\bibinfo  {journal} {J. Magn. Reson.}\ }\textbf {\bibinfo {volume} {172}},\ \bibinfo {pages} {296} (\bibinfo {year} {2005})}\BibitemShut {NoStop}%
\end{thebibliography}%
\end{document}